\documentclass[letterpaper, 10 pt]{ieeeconf}  
\IEEEoverridecommandlockouts 
\usepackage{geometry}
\newgeometry{top=0.75in, bottom=0.75in, right=0.75in, left=0.75in}

\newcommand{\bigzono}[1]{\Big\langle #1 \Big\rangle}

\newcommand{\zono}[1]{\langle #1 \rangle}
\usepackage{packagearticleieee}
\usepackage{amssymb}

\newtheorem{proposition}{Proposition}

\usepackage{cite}
\usepackage{booktabs}

 \usepackage{algorithm} 
 \usepackage{algpseudocode}
 \usepackage[utf8]{inputenc}
\usepackage{graphicx}
\usepackage{subcaption}

\usepackage{enumitem}

\makeatletter
\let\NAT@parse\undefined
\makeatother
\usepackage{url}
\usepackage[colorlinks=true,linkcolor=blue!80!black,citecolor=blue!80!black,urlcolor=blue!80!black]{hyperref}
\newcommand{\Rn}{\R^n}
\newcommand{\operator}[1]{{\normalfont \texttt{#1}}}

\newcommand{\CE}{R}
\newcommand{\Comp}{\textit{Complexity: }}

\DeclareMathSymbol{\shortminus}{\mathbin}{AMSa}{"39}

\newcommand{\id}{\mathsf{id}}

\def\tr{\text{tr}}

\title{\LARGE\bf Data-Driven Nonconvex Reachability Analysis \\ using Exact Set Propagation}

\author{Zhen Zhang$^1$, M. Umar B. Niazi$^2$, Michelle S. Chong$^3$, Karl H. Johansson$^2$, and Amr Alanwar$^1$
\thanks{$^1$ School of Computation, Information and Technology, Technical University of Munich, Germany. (Email: $\{$zhenzhang.zhang, alanwar$\}$@tum.de)}
\thanks{$^2$ Division of Decision and Control Systems, Digital Futures, School of Electrical Engineering and Computer Science, KTH Royal Institute of Technology, Stockholm, Sweden  (Email: $\{$mubniazi, kallej$\}$@kth.se)}
\thanks{$^3$ Department of Mechanical Engineering, Eindhoven University of Technology, the Netherlands (Email: m.s.t.chong@tue.nl)}
}

\begin{document}
\maketitle
\thispagestyle{empty}
\pagestyle{empty}

\begin{abstract}

This paper studies deterministic data-driven reachability analysis for dynamical systems with unknown dynamics and nonconvex reachable sets. Existing deterministic data-driven approaches typically employ zonotopic set representations, for which the multiplication between a zonotopic model set and a zonotopic state set cannot be represented algebraically exactly, thereby necessitating over-approximation steps in reachable-set propagation. To remove this structural source of conservatism, we introduce constrained polynomial matrix zonotopes (CPMZs) to represent data-consistent model sets, and show that the multiplication between a CPMZ model set and a constrained polynomial zonotope (CPZ) state set admits an algebraically exact CPZ representation. This property enables set propagation entirely within the CPZ representation, thereby avoiding propagation-induced over-approximation and even retaining the ability to represent nonconvex reachable sets. Moreover, we develop set-theoretic results that enable the intersection of data-consistent model sets as new data become available, yielding the proposed online refinement scheme that progressively tightens the data-consistent model set and, in turn, the resulting reachable set. Beyond linear systems, we extend the proposed framework to polynomial dynamics and develop additional set-theoretic results that enable both model-based and data-driven reachability analysis within the same algebraic representation. By deriving algebraically exact CPZ representations for monomials and their compositions, reachable-set propagation can be carried out directly at the set level without resorting to interval arithmetic or relaxation-based bounding techniques. Numerical examples for both linear and polynomial systems demonstrate a significant reduction in conservatism compared to state-of-the-art deterministic data-driven reachability methods.

\end{abstract}

\section{Introduction}

Reachability analysis characterizes the set of all states that a dynamical system can reach from a given initial set under admissible inputs and uncertainty, and constitutes a fundamental tool for safety verification and performance analysis of dynamical systems \cite{Althoff2010PhD, Kuhn1998}. In safety-critical applications such as autonomous vehicles, robotics, and aerospace systems, reachability guarantees are often mandated by regulatory standards, where constraint violations are unacceptable and worst-case behavior must be rigorously accounted for \cite{rierson2017developing}.

Classical reachability methods typically rely on accurate mathematical models of system dynamics, which are frequently unavailable or costly to obtain for complex systems. This limitation has motivated the development of data-driven reachability analysis, which leverages measured input--state data to infer reachable sets without explicit model identification \cite{alanwar2021data,Alanwar2023Datadriven}.

Data-driven reachability methods can be broadly classified into stochastic and deterministic frameworks. Stochastic approaches characterize reachable sets in a probabilistic sense and rely on distributional assumptions or statistical guarantees, including kernel distribution embeddings \cite{thorpe2019model,thorpe2021sreachtools}, empirical inverse Christoffel functions with PAC guarantees \cite{devonport2023data}, Gaussian process state-space models \cite{griffioen2023data}, and statistical reachability under distribution shifts \cite{hashemi2024statistical}. Data-driven reachability has also been investigated for human-in-the-loop systems and learning-based settings using Gaussian mixture models, optimized disturbance bounds, and neural network approximations \cite{govindarajan2017data,choi2023data,sivaramakrishnan2024stochastic}. While powerful, these approaches provide probabilistic guarantees and may not be suitable in applications where deterministic safety guarantees are required.

In contrast, deterministic data-driven reachability methods aim to compute set-valued enclosures that are guaranteed to contain all trajectories consistent with bounded disturbances and measured data \cite{alanwar2021data,Alanwar2023Datadriven}. Such guarantees are particularly important when probabilistic constraint violations cannot be tolerated, for instance, in collision avoidance and safety certification \cite{rierson2017developing}. Moreover, when data is limited or uncertainty distributions are unknown, deterministic formulations avoid potentially incorrect distributional assumptions and provide robustness against worst-case scenarios. Within this paradigm, zonotopic set representations and their constrained variants have become widely adopted due to their favorable computational properties and their ability to encode data consistency and prior knowledge \cite{scott2016constrained,alanwar2021data,Alanwar2023Datadriven,alanwar2022enhancing,farjadnia2024robust}.

A central challenge in deterministic reachability analysis is the accurate representation and propagation of nonconvex reachable sets. Nonconvexity naturally arises from obstacle-induced initial sets, logical constraints, and nonlinear system behavior. While convex over-approximations offer computational simplicity, they may introduce significant conservatism and obscure critical geometric features of the reachable set \cite{boyd2004convex,Alanwar2023Datadriven}. Although recent works have addressed aspects of nonconvexity through nonconvex optimization, data-driven set estimation, and reachability-based performance guarantees \cite{dietrich2024nonconvex,wang2023data,djeumou2022fly}, accurately propagating nonconvex reachable sets without resorting to conservative relaxations remains challenging, particularly in data-driven settings.

This difficulty becomes structural already for linear time-invariant (LTI) systems. In deterministic data-driven reachability, reachable sets are obtained by multiplying a set of system matrices consistent with data and bounded disturbances with a set-valued state-input bundle \cite{alanwar2021data,Alanwar2023Datadriven}. Existing approaches typically represent data-consistent models using matrix zonotopes or constrained matrix zonotopes and reachable sets using constrained zonotopes. However, constrained zonotopes do not admit an algebraically exact representation for the multiplication between constrained matrix zonotopes and constrained zonotopes. Hence, reachability propagation necessarily relies on outer approximations, which accumulate conservatism over time. Moreover, since constrained zonotopes represent convex sets, they cannot capture nonconvex reachable-set geometries, even when such nonconvexity arises naturally from the underlying uncertainty.

For polynomial dynamics, the limitations are even more pronounced. In the model-based setting, reachability analysis for polynomial systems commonly relies on bounding the image of sets under polynomial maps using Bernstein expansions or sum-of-squares optimization \cite{dang2012reachability,dreossi2017reachability,lin2022reachable,jones2019using}. While these methods provide formal guarantees, they inherently yield conservative outer or inner approximations due to polynomial bounding or certificate relaxation. Similarly, polynomial-zonotope-based nonlinear reachability methods typically employ nonlinear mapping followed by order reduction to maintain tractability, resulting in conservative flowpipe enclosures in general \cite{kochdumper2020sparse}. In the data-driven setting, deterministic frameworks for polynomial systems infer sets of data-consistent models and propagate reachable sets using set arithmetic, which again yields guaranteed over-approximations \cite{Alanwar2023Datadriven,park2024data}. Consequently, existing approaches generally do not provide algebraically exact set-level propagation within zonotopic representations for polynomial systems, neither in the model-based nor in the data-driven setting.

In this paper, we address the above limitations by developing an exact algebraic framework for deterministic data-driven reachability. 
The key idea is to exploit the richer algebraic structure of constrained polynomial zonotopes (CPZs) \cite{kochdumper2023constrained}, which generalize classical zonotopic representations by encoding polynomial dependencies. 
Building on this representation, we introduce constrained polynomial matrix zonotopes (CPMZs) to represent data-consistent model sets and show that the multiplication of a CPMZ with a CPZ admits an algebraically exact CPZ representation, explicitly incorporating constraints. 
In contrast to existing exact multiplication results for unconstrained polynomial zonotopes and matrix zonotopes \cite{kochdumper2020sparse,luo2023reachability}, the proposed construction enables exact propagation of nonconvex reachable sets in data-driven settings without relaxation.

Beyond exact propagation for a fixed dataset, we further develop an iterative set refinement mechanism that exploits the streaming availability of data. As new input--state trajectories are collected, successive data-consistent model sets can be intersected to yield progressively tighter enclosures of the true system dynamics. By integrating this refinement mechanism with the proposed exact propagation operators, the resulting reachable sets can be tightened online, reducing conservatism as more data become available.

Finally, the proposed algebraic framework extends naturally to polynomial systems. By leveraging exact multiplication and Cartesian product operations within the class of CPZs, we establish an exact reachability framework for polynomial dynamics in both the model-based and data-driven settings. In particular, reachable sets are propagated by evaluating monomial dynamics directly at the set level, avoiding interval arithmetic or relaxation-based over-approximations that are inherent in existing approaches \cite{conf:polyDissipat,alanwar2025polynomial,park2024data,althoff2015introduction}.

The main technical contributions of this paper are as follows:
\begin{enumerate}






  \item A new set representation, termed the constrained polynomial matrix zonotope (CPMZ), is introduced together with essential operations for reachability analysis and propagation.


 \item An algebraically exact CPZ representation is derived for the multiplication between a CPMZ data-consistent model set and a CPZ state set, and it is used to obtain a reachability recursion that propagates reachable sets through uncertain model sets without introducing additional enclosure beyond the assumed uncertainty bounds, while preserving nonconvex geometry.


\item The CPMZ framework is shown to naturally accommodate nonconvex zonotopic disturbance sets, enabling deterministic data-driven reachability under richer noise descriptions without conservative convex over-approximations.


\item An online model-set refinement mechanism is proposed that constructs a data-consistent model set from newly acquired input--state data and intersects it with the current model set, yielding a nested sequence of data-consistent model sets that tightens monotonically as more data become available.

\item For model-based reachability analysis of polynomial systems, we develop an algebraically exact set propagation method within the CPZ representation, which evaluates polynomial dynamics directly at the set level and avoids interval arithmetic and relaxation-based bounds.



\item CPMZ-based data-consistent model-set learning is combined with algebraically exact CPZ propagation to enable deterministic data-driven reachability analysis for polynomial systems, including online model-set refinement via set intersection.

\end{enumerate}

Overall, the proposed framework provides a unified algebraic foundation for deterministic data-driven reachability analysis, enabling exact set-based propagation of nonconvex reachable sets under model uncertainty and extending from linear to polynomial dynamical systems.

The remainder of this paper is organized as follows.
Section~\ref{sec:preliminaries} introduces the preliminaries and formulates the data-driven reachability problem.
Section~\ref{sec:DDSV} presents the proposed exact data-driven reachability framework for linear systems with nonconvex set representations.
Section~\ref{sec:reachnonlinear} extends the proposed framework to polynomial systems and establishes algebraically exact set propagation within the CPZ representation in both the model-based and data-driven settings.
Section~\ref{sec:numerical-simulations} provides numerical examples that illustrate the effectiveness of the proposed approach.
Finally, Section~\ref{sec:conclusion} concludes the paper and discusses future research directions.


\section{Preliminaries and Problem Formulation}\label{sec:preliminaries}

\subsection{Notations}
The sets of real and natural numbers are denoted by $\R$ and $\N$, respectively, with $\N_0 = \N\cup \{0\}$. 
We denote the matrix of zeros and ones of size $m \times n$ by $0_{m \times n}$ and $1_{m \times n}$, respectively, and the identity matrix by $I_n\in\R^{n\times n}$, where the subscripts are sometimes omitted to avoid clutter.
For a matrix $A$, $A^\top$ denotes the transpose and $A^\dagger$ the Moore–Penrose pseudoinverse.
By $A_{(i,j)}$, we denote the $(i,j)$-th entry and by $A_{(\cdot,j)}$, the $j$-th column. 
A slight abuse of notation is adopted for indexed matrices (e.g., $A_n$), where we use $A_n^{(i,j)}$ and $A_n^{(\cdot,j)}$.
For a vector or list $v$, $v_{(i)}$ denotes the $i$-th element, and $v_{(p_1:p_2)} \triangleq (v_{(p_1)}, \dots, v_{(p_2)})$ its restriction. 
The vectorization of a matrix $A \in \R^{n \times m}$ is denoted by $\mathrm{vec}(A) \in \R^{nm \times 1}$.
For a set $\mathcal{H} = \{1,2,\dots,{|\mathcal{H}|}\}$ with $|\mathcal{H}|$ denoting its cardinality, $A_{(\cdot,[1:|\mathcal{H}|])}$ denotes the matrix $[A_{(\cdot,1)},\dots, A_{(\cdot,|\mathcal{H}|)}]$.


\begin{table}[b]
    \centering
    \begin{tabular}{ll}
        \toprule
        CZ & Constrained Zonotope \\
        CPZ & Constrained Polynomial Zonotope \\
        MZ & Matrix Zonotope \\
        CMZ & Constrained Matrix Zonotope \\
        CPMZ & Constrained Polynomial Matrix Zonotope \\
        \bottomrule
    \end{tabular}
    \caption{List of acronyms}
    \label{tab:acronyms}
\end{table}

\subsection{Problem Statement}

We consider a system in discrete-time $k\in\N_0$ with an unknown system model:
\begin{equation} \label{eq:model-linear}
    x_{(k)} = f(x_{(k-1)},u_{(k-1)}) + w_{(k)},
\end{equation}
where $x_{(k)} \in \R^{n_x}$ is the state at time $k$, $f:\mathbb{R}^{n_x}\times\mathbb{R}^{n_u} \rightarrow \mathbb{R}^{n_x}$ a twice differentiable unknown function, $u_{(k)} \in \R^{n_u}$ is the control input, and $w_{(k)} \in \R^{n_x}$ is the unknown noise.
Reachability analysis computes the set containing all possible states $x_{(k)}$ that can be reached from a compact set $\mathcal{X}_0$ containing initial states given a compact set $\mathcal{U}_k\ni u_{(k)}$ of feasible inputs and a compact set $\mathcal{Z}_w\ni w_{(k)}$ of possible noise.

\begin{definition}
The \textit{exact reachable set}
\begin{align} 
    \label{eq:R}
    \mathcal{R}_{N} \triangleq \big\{ x_{(N)} \in \R^{n_x} \mid \ & x_{(k+1)} = f(x_{(k)} , u_{(k)}) + w_{(k)}, \nonumber \\&
    x_{(0)} \in \mathcal{X}_0,
    u_{(k)} \in \mathcal{U}_k, w_{(k)} \in \mathcal{Z}_w,\nonumber\\ &  \forall k \in \{0,...,N{-}1\} \big\}
\end{align}
is the set of all states that can be reached after $N$ time steps starting from the initial set $\mathcal{X}_0$
subject to inputs $u_{(k)} \in \mathcal{U}_k$ and noise $w_{(k)} \in \mathcal{Z}_w$, for $ k \in \{0,1,\dots,N-1\}$.
\hfill $\lrcorner$
\end{definition}

\begin{remark}[Notion of exactness]\label{rem:exactness}
Definition~1 uses the term exact reachable set to denote the standard reachable set induced by the system dynamics, inputs, and disturbances.
Throughout this paper, exact refers only to the algebraic exactness of set operations within the chosen set representation.
A set operation is algebraically exact if the resulting set representation equals the exact image of the operand sets under the corresponding mapping, without introducing additional relaxation or enclosure beyond the assumed uncertainty bounds.
This notion of exactness is representation-dependent and does not refer to the exact computation of the true reachable set.
Accordingly, the operators $\boxtimes$, $\boxplus$, and $\otimes$ in Section~\ref{sec:DDSV} are algebraically exact in this sense.
\end{remark}

Our objective is to compute data-driven reachable sets for the system~\eqref{eq:model-linear} using input--state data, without explicit knowledge of the system model. In contrast to model-based approaches, the system matrices are assumed to be unknown and are only constrained by measured data and bounded disturbances.

Beyond estimating a single reachable set, we seek a reachability framework that explicitly accounts for model uncertainty and the availability of newly acquired data. In particular, the goal is to construct reachable sets that (i) rigorously enclose all trajectories consistent with the data and noise bounds, (ii) preserve nonconvex geometric features of the initial set $\mathcal{X}_{0}$, and (iii) can be iteratively refined as additional input--state data become available.

More specifically, we aim to develop a data-driven reachability method that enables exact set-based propagation of nonconvex reachable sets under a set of data-consistent system models
, while allowing both the model set and the corresponding reachable sets to be progressively tightened through online set refinement.

The proposed framework is applies to the following classes of systems:
\begin{enumerate}
    \item \emph{Linear time-invariant (LTI) systems} (Section~\ref{sec:DDSV}), described by
    \begin{equation} \label{eq:model-linear}
        x_{(k)} = \Phi_{\mathrm{tr}} x_{(k-1)} + \Gamma_{\mathrm{tr}} u_{(k-1)} + w_{(k)},
    \end{equation}
    where $\begin{bmatrix} \Phi_{\mathrm{tr}} & \Gamma_{\mathrm{tr}} \end{bmatrix}$ denotes the true, but unknown, system model.
    \item \emph{Polynomial systems} (Section~\ref{sec:reachnonlinear}), for which exact set-based propagation of reachable sets is achieved in both the model-based and data-driven settings.
\end{enumerate}

\subsection{Set Representations}
For our data-driven reachability analysis, we define the required set representations below. 

\begin{definition}
\label{def:conZonotope}
Given an offset $c\in\Rn$, generator $G\in\R^{n \times h}$, and constraints $A \in \R^{n_c \times h}$ and $b \in \R^{n_c}$, a \textit{constrained zonotope} (CZ) is a compact set in $\Rn$ defined as (see \cite{scott2016constrained})
\begin{equation}
    \label{eq:con-zon}
    \mathcal{C}{=}\bigg\{c+ \sum_{k=1}^{h} \alpha_{(k)} G_{(\cdot,k)} \biggm|  \sum_{k=1}^h \alpha_{(k)} A_{(\cdot,k)}=b,  \alpha_{(k)}\in[\shortminus 1,1] \bigg \}.
\end{equation}
Furthermore, for reasons that will become apparent later, we associate an identifier $\id \in  \N^{1 \times h }$ with $\mathcal{C}$ for identifying the factors $\alpha_{(1)},\dots,\alpha_{(h)}$ in \eqref{eq:con-zon}.
To describe a CZ, we use a shorthand notation $\mathcal{C} = \zono{c,G,A,b,\id}_\text{CZ}$. 
\hfill $\lrcorner$
\end{definition}

Note that zonotopes are a special case of CZs, where constraints are empty \cite{kochdumper2023constrained}, i.e., $A=[~]$ and $B=[~]$.

\begin{definition}
Given an offset $c\in\Rn$, generator $G\in\R^{n \times h}$, exponent $E \in \N_0^{p \times h}$, and constraints $A \in \R^{n_c \times q}$, $b \in \R^{n_c}$, and $\CE \in \N_0^{p \times q}$, a \textit{constrained polynomial zonotope} (CPZ) is a compact set in $\Rn$ defined as (see \cite{kochdumper2023constrained})
\begin{multline}
    \label{eq:con-poly-zono}
    \mathcal{P} = \bigg \{ c + \sum_{i=1}^{h} \bigg( \prod_{k=1}^p \alpha_{(k)}^{E_{(k,i)}} \bigg) G_{(\cdot,i)} ~ \bigg | \\ 
    \sum_{i=1}^{q} \bigg( \prod_{k=1}^p \alpha_{(k)}^{\CE_{(k,i)}} \bigg) A_{(\cdot,i)} = b, \alpha_{(k)} \in [\shortminus 1,1]   \bigg \}.
\end{multline}

Similar to CZs, we associate an identifier $\id \in  \N^{1 \times p }$ with $\mathcal{P}$ for identifying the factors $\alpha_{(1)},\dots,\alpha_{(p)}$ in \eqref{eq:con-poly-zono}. 
We denote a CPZ as $\mathcal{P} = \zono{c,G,E,A,b,R,\id}_\text{CPZ}$.
\hfill $\lrcorner$
\end{definition}

\begin{example}
To understand the role of $\id$, consider
\begin{equation}
\label{eq:barp1}
    \mathcal{P}_1=\Biggl\langle    
    \begingroup
    \setlength\arraycolsep{1pt}
    \begin{bmatrix} 0 \\ 2 \\ 1 \end{bmatrix},
    \begin{bmatrix} 0 & 1 \\ 3 & 2 \\ 1 & 5 \end{bmatrix},
    \begin{bmatrix} 4 & 1 \\ 0 & 2 \end{bmatrix},
    \begin{bmatrix} 1 & 2 \\ 0 & 0 \\ 3 & 4 \end{bmatrix},
    \begin{bmatrix} 2 \\ 0 \\ 2 \end{bmatrix},
    \begin{bmatrix} 4 & 2 \\ 0 & 2 \end{bmatrix},
    \begin{bmatrix} 1 & 2 \end{bmatrix}
    \endgroup
    \Biggr\rangle_\text{CPZ}
\end{equation}
which describes the following CPZ
\begin{multline}
    \mathcal{P}_1 = \Bigg\{ 
    \begin{bmatrix} 0 \\ 2 \\ 1 \end{bmatrix} 
    +\begin{bmatrix} 0 \\ 3 \\ 1 \end{bmatrix} 
    \alpha_{(1)}^4 
    +\begin{bmatrix} 1 \\ 2 \\ 5 \end{bmatrix} 
    \alpha_{(1)} \alpha_{(2)}^2  \Biggm| \\ 
    \begin{bmatrix} 1 \\ 0 \\ 3 \end{bmatrix}
    \alpha_{(1)}^4
    +\begin{bmatrix} 2 \\ 0 \\ 4 \end{bmatrix}
    \alpha_{(1)}^2\alpha_{(2)}^2
    =\begin{bmatrix} 2 \\ 0 \\ 2 \end{bmatrix},
    \alpha_{(1)},\alpha_{(2)} \in [\shortminus 1, 1]\Bigg\}.
\end{multline}
Here, $\id = \begin{bmatrix} 1 & 2\end{bmatrix}$ identifies the dependent factor $\alpha_{(1)}$ with $\id_{(1)}=1$ and $\alpha_{(2)}$ with $\id_{(2)}=2$.
\hfill $\lrcorner$
\end{example}

\begin{definition} \label{def:conmatzonotopes}  
Given an offset $C \in \R^{m \times n}$, generators $G=\begin{bmatrix} G_{(1)}\dots G_{(\gamma)}\end{bmatrix} \in \R^{m \times (n\gamma)}$, constraints $A=\begin{bmatrix} A_{(1)} \dots A_{(\gamma)}\end{bmatrix} \in \R^{n_c \times (n_a\gamma)}$, and $B \in \R^{n_c \times n_a}$, a \textit{constrained matrix zonotope} (CMZ) is a compact set in $\R^{m\times n}$ defined as (see \cite{Alanwar2023Datadriven})
\begin{equation}
    \label{eq:con-mat-zono}
    \mathcal{N} =\Big\{ C + \sum_{k=1}^{\gamma} \alpha_{(k)} G_{(k)} \Bigm|
    \sum_{k=1}^{\gamma} \alpha_{(k)} A_{(k)} = B,  \alpha_{(k)} \in [\shortminus 1,1]  \Big\} .
\end{equation}
We associate an identifier $\id \in  \N^{1 \times p }$ with $\mathcal{N}$ for identifying the factors $\alpha_{(1)},\dots,\alpha_{(p)}$.
Denote $\mathcal{N} = \zono{C,G,A,B,\id}_\text{CMZ}$. 
\hfill $\lrcorner$
\end{definition}

Note that matrix zonotopes are a special case of CMZs, denoted by $\mathcal{N}=\zono{C,G,[~],[~],\id}_\text{CMZ}$. 

\begin{definition}
\label{def:conmatpolyzonotopes}  
Given an offset $C\in \R^{m \times n}$, generators $G=\begin{bmatrix} G_{(1)} \dots G_{(\gamma)} \end{bmatrix} \in \R^{m \times (n  \gamma)}$, exponent $E \in \N_0^{p \times \gamma}$, constraints $A=\begin{bmatrix} A_{(1)} \dots A_{(\gamma)} \end{bmatrix} \in \R^{n_c \times (n_a \gamma)}$, $B \in \R^{n_c \times n_a}$, and $\CE \in \N_0^{p \times \gamma}$, a \textit{constrained polynomial matrix zonotope} (CPMZ) is 
\begin{multline}
    \label{eq:con-poly-mat-zono}
    \mathcal{Y} = \Big\{ C + \sum_{i=1}^{\gamma} \bigg( \prod_{k=1}^p \alpha_{(k)}^{E_{(k,i)}} \bigg) \, G_{(i)} \Bigm| \\
    \sum_{i=1}^{\gamma} \bigg( \prod_{k=1}^p \alpha_{(k)}^{\CE_{(k,i)}} \bigg) A_{(i)} = B \,,   \alpha_{(k)} \in [\shortminus 1,1] \Big\} \; .
\end{multline}
We associate an identifier $\id \in \N^{1\times p}$ with $\mathcal{Y}$ for identifying the factors $\alpha_{(1)},\dots,\alpha_{(p)}$.
Furthermore, we use the shorthand notation $\mathcal{Y} = \zono{C,G,E,A,B,R,\id}_\text{CPMZ}$. 
\hfill $\lrcorner$
\end{definition}

To perform operations between two CPZs, the \operator{mergeID} operator is needed to make them compatible with each other.

\begin{proposition}[\operator{mergeID} \cite{kochdumper2020sparse}] \label{prop:mergeID}  
Given two CPZs
\begin{align*}
\mathcal{P}_1 &= \zono{c_{1}, G_{1}, E_{1}, A_{1}, b_{1}, R_{1},  \id_{1}}_\text{CPZ} \\
\mathcal{P}_2 &= \zono{c_{2}, G_{2}, E_{2}, A_{2}, b_{2}, R_{2}, \id_{2}}_\text{CPZ}
\end{align*}
the \operator{mergeID} operator returns two adjusted CPZs that are equivalent to $\mathcal{P}_1$ and $\mathcal{P}_2$:
\begin{align}
      \operator{mergeID}&(\mathcal{P}_1,\mathcal{P}_2) =& \nonumber \\
       & \qquad\big \{ \underbrace{\langle C_{{1}}, G_{{1}}, \overline{E}_{1}, A_{{1}}, B_{{1}}, \overline{R}_{1}, \overline{\id}\rangle_{\text{CPZ}}}_{\bar{_1}}, \nonumber \\
  &\qquad\underbrace{\langle c_{2}, G_{2}, \overline{E}_{2}, A_{2}, b_{2}, \overline{R}_{2},\overline{\id}\rangle_{\text{CPZ}}}_{\bar{\mathcal{P}}_2} \big \}.  
    \end{align}
with $\overline{\id}= \begin{bmatrix} \id_{1} & \id_2^{(\cdot,\mathcal{H})} \end{bmatrix}$, $\mathcal{H} = \left\{ i~ |~ \id_{2}^{(i)} \not\in \id_{1} \right\}$, and
\begin{align*}
    \overline{E}_{1} &= \begin{bmatrix} E_{1} \\ 0_{|\mathcal{H}|\times h_1} \end{bmatrix} \in \R^{a \times h_{1}}, \qquad
    \overline{R}_{1} = \begin{bmatrix} R_{1} \\ 0_{|\mathcal{H}|\times q_{1}} \end{bmatrix} \in \R^{a \times q_{1}} \\
    \overline{E}_{2}^{(i,\cdot)} &= \begin{cases} E_{2}^{(j,\cdot)}, & \mathrm{if} ~ \exists j~\overline{\id}_{(i)} = \id_{2}^{(j)} \\ 
    0_{1\times h_{2}}, & \mathrm{otherwise} 
    \end{cases}\\
    \overline{R}_{2}^{(i,\cdot)} &= \begin{cases} R_{2}^{(j,\cdot)}, & \mathrm{if} ~ \exists j~\overline{\id}_{(i)} = \id_{2}^{(j)} \\ 
    0_{1\times q_{2}}, & \mathrm{otherwise}
    \end{cases}
\end{align*}
where $i = 1, \dots, a$ with $a = |\mathcal{H}|+p_{1}$ for $\id_{1} \in  \N^{1 \times p_{1} }$.
\hfill $\lrcorner$
\end{proposition}

\begin{example}
    Consider a CPZ
\begin{displaymath}
    {\mathcal{P}}_2 = \Biggl \langle
    \begingroup
    \setlength\arraycolsep{2pt}
    \begin{bmatrix}
    3\\3\\4
    \end{bmatrix},\begin{bmatrix}
    2& 2\\3 & 0\\1 & 4
    \end{bmatrix},\begin{bmatrix}
    3 & 2\\3 & 0
    \end{bmatrix},\begin{bmatrix}
    1& 3\\2 & 4
    \end{bmatrix},\begin{bmatrix}
    2\\ 5
    \end{bmatrix}, 
    \begin{bmatrix}
    2 & 0\\2 & 3
    \end{bmatrix},\begin{bmatrix}
    1  & 3
    \end{bmatrix}
    \endgroup
    \Biggr \rangle_\text{CPZ}
\end{displaymath}
describing the following set
\begin{multline}
 {\mathcal{P}}_2 = \Bigg\{ \begin{bmatrix}
 3\\3\\4
\end{bmatrix} + \begin{bmatrix}
 2\\3\\1
\end{bmatrix} \alpha_{(1)}^3\alpha_{(2)}^3 + \begin{bmatrix}
 2\\0\\4
\end{bmatrix} \alpha_{(1)}^2  \Biggm| \\
\alpha_{(1)}^2\alpha_{(2)}^2\begin{bmatrix}
 1\\2
\end{bmatrix} 
+ \alpha_{(2)}^3\begin{bmatrix}
 3\\4
\end{bmatrix}=\begin{bmatrix}
 2\\5
\end{bmatrix}, \alpha_{(1)},\alpha_{(2)} \in[-1,1]\Bigg\}
\end{multline}
where $\id = \begin{bmatrix} 1 & 3\end{bmatrix}$ identifies the dependent factor $\alpha_{(1)}$ with $\id_{(1)}=1$ and $\alpha_{(2)}$ with $\id_{(2)}=3$. If we apply the operator \operator{mergeID}(${\mathcal{P}}_1, {\mathcal{P}}_2$) where ${\mathcal{P}}_1$ is defined in \eqref{eq:barp1}, we get the following sets with common identifiers.
\begin{align*}
    \bar{\mathcal{P}}_1 &= \Biggl \langle
    \begingroup
    \setlength\arraycolsep{2pt}
    \begin{bmatrix}
    0\\2\\1
    \end{bmatrix},\begin{bmatrix}
    0 & 1\\3 & 2\\1 & 5
    \end{bmatrix},\begin{bmatrix}
    4 & 1\\0 & 2\\0 & 0 
    \end{bmatrix},\begin{bmatrix}
    1 & 2\\0 & 0\\3 & 4
    \end{bmatrix},\begin{bmatrix}
    2\\0\\2 
    \end{bmatrix},\begin{bmatrix}
    4 & 2\\0 & 2\\0 & 0 
    \end{bmatrix},\begin{bmatrix}
    1 &2 & 3
    \end{bmatrix} 
    \endgroup
    \Biggr \rangle \nonumber
\\
    \bar{\mathcal{P}}_2 &= \Biggl \langle
    \begingroup
    \setlength\arraycolsep{2pt}
    \begin{bmatrix}
    3\\3\\4
    \end{bmatrix},\begin{bmatrix}
    2& 2\\3 & 0\\1 & 4
    \end{bmatrix},\begin{bmatrix}
    3 & 2\\0 & 0\\ 3& 0
    \end{bmatrix},\begin{bmatrix}
    1 & 3\\2 & 4
    \end{bmatrix},\begin{bmatrix}
    2\\5
    \end{bmatrix},\begin{bmatrix}
    2 & 0\\0 & 0\\2 & 3 
    \end{bmatrix},\begin{bmatrix}
    1 &2 &3
    \end{bmatrix}  
    \endgroup
    \Biggr \rangle.
\end{align*}
\hfill $\lrcorner$
\end{example}

\subsection{Assumptions}
We assume access to $n_T \in \N$ input-state trajectories of length $T_i + 1$, with inputs $\{ u_{(k)}\}^{T_i-1}_{k=0}$ and states $\{ x_{(k)}\}^{T_i}_{k=0}$ for $i\in\{0,\dots,n_T-1\}$. For simplicity, we consider the offline data from a single trajectory ($n_T = 1$) of length $T_0$ and organize the input and
noisy state data into matrices:
\begin{equation}
    \label{eq:offline-data}
    \begin{array}{c}
    X_0^+ = \begin{bmatrix}  x_{(1)} & x_{(2)} & \dots & x_{(T_0)}\end{bmatrix},\\[2pt]
    X_0^- = \begin{bmatrix}  x_{(0)}& x_{(1)} & \dots & x_{(T_0-1)}\end{bmatrix}, \\[2pt]
    U_0^- = \begin{bmatrix} u_{(0)} &  u_{(1)}  & \dots & u_{(T_0-1)} \end{bmatrix}.
    \end{array}  
\end{equation} 
Let $D_0^-= 
\begingroup
\setlength\arraycolsep{2pt}
\begin{bmatrix} X_0^{-\top} & U_0^{-\top} \end{bmatrix}^\top
\endgroup$ and $D_0= 
\begingroup
\setlength\arraycolsep{2pt}
\begin{bmatrix} X_0^{+\top} & X_0^{-\top} &U_0^{-\top}\end{bmatrix}^\top
\endgroup$. 

\begin{assumption}\label{ass:zon-noise}
    For all $k\in \Z_{\geq 0}$, the noise $w_{(k)}$ is bounded by a known CPZ $\mathcal{Z}_w$, which includes the origin.
    \hfill $\lrcorner$
\end{assumption}
\begin{assumption}\label{ass:rank_D}
    The data matrix $D_0^-$ has full row rank, i.e., $\mathrm{rank}(D_0^-) = n_x + n_u$.
    \hfill $\lrcorner$
\end{assumption}

Let all system matrices $\begin{bmatrix} \Phi & \Gamma \end{bmatrix}$ of \eqref{eq:model-linear} consistent with the data $D_-$ be
\begin{equation}
    \mathcal{N}_{\Sigma} = \Big\{ 
    \begingroup
    \setlength\arraycolsep{3pt}
    \begin{bmatrix} \Phi & \Gamma \end{bmatrix}
    \endgroup
    \Bigm| X_0^+ = \Phi X_0^- + \Gamma U_0^- + W_0^-, W_0^- \in \mathcal{M}_w \Big\}. \label{eq:Nsig}
\end{equation}
By definition, $\begin{bmatrix} \Phi_\tr & \Gamma_\tr \end{bmatrix} \in \mathcal{N}_{\Sigma}$.

\section{Data-driven Reachability Analysis with Nonconvex Set Representations}\label{sec:DDSV}

Given system \eqref{eq:model-linear} with unknown matrices $\Phi_\tr$ and $\Gamma_\tr$, the presence of noise implies multiple feasible matrices $\begin{bmatrix} \Phi & \Gamma \end{bmatrix}$ are consistent with the observed data. 
For the reachability analysis to be reliable, one must incorporate all such data-consistent models.
To this end, we extend the offline approach proposed by \cite{Alanwar2023Datadriven} to continue refining the estimated set of models by leveraging online data.

Our proposed 
data-driven reachability analysis
algorithm 
comprises two steps: 
first, determining a set of models consistent with observed data given bounded uncertainties; and second, using this model set to compute the region in state space containing all possible true states.
Using offline data \eqref{eq:offline-data}, we compute an initial set of models represented by a CPMZ. Then, in the online phase, we refine this CPMZ iteratively by collecting online input-state trajectories.
For computing the reachable set as a CPZ in the online phase, we propose an exact multiplication method for the time update to multiply the CPMZ representing a set of models with the CPZ representing the previous reachable set.

\begin{figure}
    \centering
    \includegraphics[width=1\linewidth]{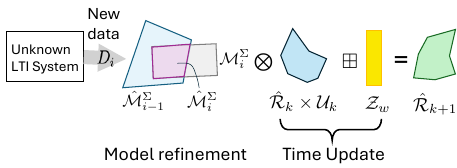}
    \caption{Data-driven reachability analysis with model refinement for LTI system. Legend: $\otimes$ denotes the exact multiplication operation and $\boxplus$ denotes the exact addition operation.}
    \label{fig:algo-cartoon}
\end{figure}
\subsection{Set of Models}


We represent the sequence of unknown noise as $\{w_{(k)}\}_{k=0}^{T_0}$. From Assumption~\ref{ass:zon-noise}, it follows that 
\[
W_0^- = \begin{bmatrix} w_{(0)} & \cdots & w_{(T_0-1)} \end{bmatrix} \in \mzon_w,
\]
where $\mathcal{M}_w =\zono{ \tilde C,\tilde{G}, \tilde  E,\tilde{A},\tilde B,\tilde R }$ is the CPMZ resulting from the concatenation of noise CPZ 
\begin{align} \mathcal{Z}_w=\zono{c_{w},\begin{bmatrix} g_{w}^{(1)}& \dots & g_{w}^{(\gamma_{w})}\end{bmatrix},E_{w},\nonumber\\\begin{bmatrix} a_{w}^{(1)}& \dots & a_{w}^{(\gamma_{w})}\end{bmatrix},B_{w},R_{w}}
\end{align}
as follows:

\begin{align*}
&\tilde{G}=\begin{bmatrix} G^{(1)}&\dots&G^{(\gamma_w)}\end{bmatrix},  \\&
\tilde{A}=\begin{bmatrix} A^{(1)}&\dots&A^{(\gamma_w)}\end{bmatrix},\\&
   \tilde  C = \begin{bmatrix}c_{w} & \dots & c_{w}\end{bmatrix}, \\&
    G^{(1+(i-1)T)} = \begin{bmatrix} g_{w}^{(i)} & 0_{n \times  (T-1)}\end{bmatrix}, \\&
    G^{(j+(i-1)T)} = \begin{bmatrix} 0_{n \times  (j-1)} &g_{w}^{(i)}  & 0_{n \times  (T-j)}\end{bmatrix}, \\&
    G^{(T+(i-1)T)} = \begin{bmatrix} 0_{n \times (T-1)}& g_{w}^{(i)}\end{bmatrix},\\&
    A^{(1+(i-1)T)} = \begin{bmatrix} a_{w}^{(i)} & 0_{n \times  (T-1)}\end{bmatrix}, \\&
    A^{(j+(i-1)T)} = \begin{bmatrix} 0_{n \times  (j-1)} &a_{w}^{(i)}  & 0_{n \times  (T-j)}\end{bmatrix}, \\&
    A^{(T+(i-1)T)} = \begin{bmatrix} 0_{n \times (T-1)}& a_{w}^{(i)}\end{bmatrix},\\&
    \tilde E = E_{w},\\&
    \tilde R = R_{w},
\end{align*}
$\forall i =\{1, \dots, \gamma_{w}\}$, $j=\{2,\dots,T-1\}$.

Using the offline data \eqref{eq:offline-data}, we compute the set of models $\mathcal{M}_0^{\Sigma}$ that is consistent with the data as follows. 

\begin{lemma}
\label{lm:sigmaM1}
Consider the offline input-state data \eqref{eq:offline-data} such that Assumption~\ref{ass:rank_D} holds and consider the CPMZ
\begin{align}
    \mathcal{M}_0^\Sigma = (X_0^+ - \mathcal{M}_w) \begin{bmatrix} 
    X_0^- \\ U_0^- 
    \end{bmatrix}^{\dagger}.
    \label{eq:zonoAB1}
\end{align} 
Then, $\begin{bmatrix} \Phi_\tr & \Gamma_\tr \end{bmatrix} \in \mathcal{M}_0^\Sigma$ and $\mathcal{N}_\Sigma \subseteq \mathcal{M}_0^\Sigma$. 
\hfill $\lrcorner$
\end{lemma}
\begin{proof}
The result follows directly from Lemma 1 in \cite{Alanwar2023Datadriven}.
\end{proof}

\begin{remark}
Requiring the full row rank in Assumption~\ref{ass:rank_D}, i.e., $\begin{bmatrix} X^{-\top} & U^{-\top}\end{bmatrix}^{\top}=n_x+n_u$, implies that there exists a right-inverse of the matrix $\left[\begin{array}{ll}
X^{-\top} & U^{-\top}
\end{array}\right]^{\top}$.
\hfill $\lrcorner$
\end{remark}

\subsection{Intersection Between Sets of models}

For the following result, we omit the identifier $\id$ of CPMZs as it is irrelevant when intersecting two sets.

\begin{proposition}[Intersection of CPMZs] 
Consider two CPMZs 
\begin{align*}
    \mathcal{Y}_1 &= \zono{C_{1},G_{1},E_{1},A_{1},B_{1},R_{1}}_\text{CPMZ}\subset \R^{n_x \times n},\\ 
    \mathcal{Y}_2 &= \zono{C_{2},G_{2},E_{2},A_{2},B_{2},R_{2}}_\text{CPMZ}\subset\R^{n_x \times n}.
\end{align*}
Then, their intersection is given by
\begin{equation}
\mathcal{Y}_1 \cap \mathcal{Y}_2=  \bigzono{ C_{1},G_{1},\hat{E},\hat{A},\hat{B},\hat{R} }_\text{CPMZ} 
\label{generalintersection}
\end{equation}
where 
\begin{align}
&\hat{E}\!=\!\left[\begin{array}{cc} E_{1}\\0_{ p_2\times \gamma_1}\end{array}\right],\\
&\hat{A}\!=\!\left[\arraycolsep=2pt\begin{array}{cccc}
\hat{A}_{1} & 0_{n_{c_1}n_{a_1} \times \gamma_2}& 0_{n_{c_1}n_{a_1} \times \gamma_1}& 0_{n_{c_1}n_{a_1} \times \gamma_2}\\
0_{n_{c_2}n_{a_2} \times \gamma_1}&\hat{A}_{2}& 0_{n_{c_2}n_{a_2} \times \gamma_1}& 0_{n_{c_2}n_{a_2} \times \gamma_2}\\ 0_{n_xn \times \gamma_1}& 0_{n_xn \times \gamma_2}&
\hat{G}_{1}&-\hat{G}_{2}
\end{array}\right], \\ &\hat{B}\!=\!\left[\begin{array}{c}
\text{vec}\big({B}_{1}\big) \\
\text{vec}\big({B}_{2}\big) \\
\text{vec}\big({C}_{2}-{C}_{1}\big)
\end{array}\right]\\ &\hat{R}\!=\!\left[\begin{array}{cccc}
R_{1} & 0_{p_1 \times \gamma_2}& E_{1}& 0_{p_1 \times \gamma_2}\\
0_{p_2 \times \gamma_1}&R_{2}& 0_{p_2 \times \gamma_1}& E_{2}
\end{array}\right]
\end{align}
with $\hat{A}_{i}= [\text{vec}({A}_{i}^{(1)}),\dots,\text{vec}({A}_{i}^{({\gamma_i})})]$ and
$\hat{G}_{i}= [\text{vec}({G}_{i}^{(1)}),\dots,\text{vec}({G}_{i}^{({\gamma_i})})]$, for $i=1,2$.
\hfill $\lrcorner$
\label{Intersection}
\end{proposition}

\begin{proof}
Let $\hat{\mathcal{Y}}$ denote the right-hand side of \eqref{generalintersection} and choose $\mathbf{Y}\in \mathcal{Y}_1 \cap \mathcal{Y}_2$. Then
\begin{align}
 \exists {\alpha}_1^{([1:p_1])}:\sum_{i=1}^{\gamma_1} \bigg( \prod_{k=1}^{p_1} \alpha_{(k)}^{\CE_{1}^{(k,i)}} \bigg) A_{1}^{(i)} = B_{1} \label{intero1}\\
 \exists {\alpha}_2^{([1:p_2])}:\sum_{i=1}^{\gamma_2} \bigg( \prod_{k=1}^{p_2} \alpha_{(k)}^{\CE_{2}^{(k,i)}} \bigg) A_{2}^{(i)} = B_{2} \label{intero2}
\end{align}
such that
 \begin{align}
\underbrace{C_{1} + \sum_{i=1}^{\gamma_1} \bigg( \prod_{k=1}^{p_1} \alpha_{(k)}^{E_{1}^{(k,i)}} \bigg)  G_{1}^{(i)}}_{\mathbf{Y} \in \mathcal{Y}_1} = \underbrace{C_{2} + \sum_{i=1}^{\gamma_2} \bigg( \prod_{k=1}^{p_2} \alpha_{(k)}^{E_{2}^{(k,i)}} \bigg)  G_{{1}}^{(i)}}_{\mathbf{Y} \in \mathcal{Y}_2} \label{intero3}
\end{align}
Let $\hat \alpha=(\alpha_1,\alpha_2)$. To account for the potential dimensional mismatch in $B_1$ and $B_2$, there exists $ {\hat\alpha}_{([1:\hat p])}$ satisfies:
\begin{align}  
 &\sum_{i=1}^{\gamma_1} \bigg( \prod_{k=1}^{p_1} \hat{\alpha}_{(k)}^{\hat\CE^{(k,i)}} \bigg) \text{vec}\big(A_{1}^{(i)}\big)= \text{vec}\big(B_{1}\big), \nonumber\\ & \sum_{i=\gamma_1+1}^{\gamma_1+\gamma_2} \bigg( \prod_{k=p_1+1}^{p_1+p_2} \hat{\alpha}_{(k)}^{\hat\CE^{(k,i)}} \bigg) \text{vec}\big(A_{2}^{(i)}\big)= \text{vec}\big(B_{2}),\nonumber\\
 &\sum_{i=1+\gamma_1+\gamma_2}^{\gamma_1+\gamma_2+\gamma_1} \bigg( \prod_{k=1}^{p_1} \hat{\alpha}_{(k)}^{\hat R^{(k,i)}} \bigg) \text{vec}\big(G_{1}^{(i)}\big)-\nonumber\\&~~~~~~~~~~~~~~~~~
          \sum_{i=1+2\gamma_1+\gamma_2}^{2\gamma_1+2\gamma_2} \bigg( \prod_{k=p_1+1}^{p_1+p_2} \alpha_{(k)}^{\hat R^{(k,i)}} \bigg)  \text{vec}\big(G_{2}^{(i)})\nonumber\\&~~~~~~~~~~~~~~~~~~~~~~~~~~~~~~~~~~~~= \text{vec}(C_{2} - C_{1})
         \label{veccon1}
\end{align}
and
\begin{align}
\mathbf{Y} =\! C_{1} + \sum_{i=1}^{\gamma_1} \bigg( \prod_{k=1}^{p_1+p_2} \hat{\alpha}_{(k)}^{\hat{E}^{(k,i)}} \bigg) G_{1}^{(i)}, \label{intersectout}
\end{align}
Thus, $\mathbf{Y} \in  \hat{\mathcal{Y}}$ and $\mathcal{Y}_1 \cap \mathcal{Y}_2 \subset \hat{\mathcal{Y}}$. 

Conversely, let $\mathbf{Y} \in  \hat{\mathcal{Y}}$. Then there exists $ {\hat\alpha}_{([1:\hat p])}$ such that
\begin{align}
\mathbf{Y} =\! C_{1} + \sum_{i=1}^{\gamma_1} \bigg( \prod_{k=1}^{\hat{p}} \hat{\alpha}_{(k)}^{\hat{E}^{(k,i)}} \bigg) G_{1}^{(i)}, \label{intersectoutend}
\end{align}
and \eqref{veccon1} holds. Partitioning $\hat \alpha$ as $\hat \alpha=(\alpha_1,\alpha_2)$, then there exists $ \hat{\alpha}_{([1:\hat p])}$ such that \eqref{intero1}, \eqref{intero2} and \eqref{intero3} holds.
These conditions imply that $\mathbf{Y} \in {\mathcal{Y}}_1$ and $\mathbf{Y} \in  {\mathcal{Y}}_2$. Thus, $\mathbf{Y} \in  {\mathcal{Y}}_1 \cap {\mathcal{Y}}_2$ and $\hat{\mathcal{Y}} \in  {\mathcal{Y}}_1 \cap {\mathcal{Y}}_2$.
\end{proof}

\Comp
The computations of $\text{vec}(B_{1})$, $\text{vec}(B_{2})$, and $\text{vec}(C_{2} - C_{1})$ require 
$\mathcal{O}(n_{c_1}n_{a_1})$, $\mathcal{O}(n_{c_2}n_{a_2})$, and $\mathcal{O}(2n_x n)$ operations, respectively.
The constructions of $\hat{A}_{1}$, $\hat{A}_{2}$, $\hat{G}_{1}$, and $\hat{G}_{2}$ involve complexities 
$\mathcal{O}(n_{c_1}n_{a_1}\gamma_1)$, $\mathcal{O}(n_{c_2}n_{a_2}\gamma_2)$, $\mathcal{O}(n_x n \gamma_1)$, and 
$\mathcal{O}(n_x n \gamma_2)$, respectively. Thus, the overall computational complexity of intersection \eqref{Intersection} is
$
\mathcal{O}\left(\gamma_1(n_{c_1}n_{a_1} + n_x n) + \gamma_2(n_{c_2}n_{a_2} + n_x n)\right).
$
\hfill $\diamond$

\begin{remark}
The intersection $\mathcal{Y}_1 \cap \mathcal{Y}_2 = \hat{\mathcal{Y}}$ in \eqref{generalintersection} yields constraint matrices $\hat{A} \in \R^{\hat{m} \times (\gamma_1 + \gamma_2)}$ and $\hat{B} \in \R^{\hat{m} \times 1}$, with $\hat{m} = n_{c_1}n_{a_1} + n_{c_2}n_{a_2} + n_x n$. This corresponds to a special case of CPMZs with $n_a = 1$, adopted to facilitate exact multiplications with CPZs in subsequent computations. When $\hat{m}/\hat{n}_c \in \N$, the vector $\hat{B}$ can be reshaped as a matrix $\hat{B} \in \R^{\hat{n}_c \times (\hat{m}/\hat{n}_c)}$ by the reshaping operator defined as follows:
\[
\operator{Convert}(\hat B\in \R^{\hat m \times 1},\hat n_c)=\hat{B}\in \R^{\hat n_{{c}} \times (\hat{m}/\hat{n}_c)}.
\]
Accordingly, the reshaped constraint matrix $\hat{A}$ lies in the space
$
\hat{A} \in \R^{\hat n_{{c}} \times \left( (\hat{m}/\hat{n}_c) \cdot (\gamma_1 + \gamma_2) \right)},
$
where
$
\hat{A}=\big[\operator{Convert}(\hat{A}^{(1)}), \dots,\operator{Convert}(\hat{A}^{(\gamma_1+\gamma_2)})\big].
$
\hfill $\lrcorner$
\end{remark}

\subsection{Iterative Refinement of Set of Models}

Similar to Lemma~\ref{lm:sigmaM1}, given any input-state data $D_i = (X_i^+, X_i^-, U_i^-)$ of system \eqref{eq:model-linear} over the interval $t \in [T_{i-1}, T_i]$, for $i\geq 1$, and that Assumption~\ref{ass:rank_D} holds, the resulting CPMZ $\mathcal{M}_i^\Sigma$ contains all matrices $\begin{bmatrix}\Phi & \Gamma\end{bmatrix}$ consistent with the data and noise bound, i.e., $\mathcal{N}_\Sigma \subseteq \mathcal{M}_i^\Sigma$.
At iteration~$i$ of the refinement phase ($i \geq 1$), we intersect the set of models 
${\mathcal{M}}_{i}^{\Sigma}$ derived from newly obtained input-state trajectories with the previous model set $\hat{\mathcal{M}}_{i-1}^{\Sigma}$ to get refined set $\hat{\mathcal{M}}_{i}^{\Sigma}$. The initial set of models and refined sets are subsequently used to compute an over-approximation of the reachable sets of the unknown system, as illustrated in Fig.~\ref{fig:algo-cartoon}.

\begin{proposition}
\label{prop:reach_lin}
Consider data $D_i = (X_i^+,X_i^-,U_i^-)$, for $i\geq 0$, with $D_0$ an initial data and $D_1,D_2,\dots$ subsequent new data obtained from system \eqref{eq:model-linear}, where each data $D_i$ satisfies Assumption~\ref{ass:rank_D} and results in a set of models $\mathcal{M}_i^\Sigma$ using Lemma~\ref{lm:sigmaM1}.
The refined set of models  
\begin{align}
&\hat{\mathcal{M}}_i^\Sigma=
\mathcal{M}_i^\Sigma \cap \hat{\mathcal{M}}_{i-1}^\Sigma, \quad i\geq 1,\quad \text{with}\quad \hat{\mathcal{M}}_0^\Sigma=
\mathcal{M}_0^\Sigma \label{cmzintersection}
\end{align}
is a CMZ and contains the true model $\begin{bmatrix} \Phi_\tr & \Gamma_\tr \end{bmatrix}$.
\end{proposition}

\begin{proof}
A straightforward consequence of Lemma~\ref{lm:sigmaM1}.
\end{proof}

\subsection{Exact Multiplication for Non-Convex Reachability}
At time $k$, as depicted in Fig.~\ref{fig:algo-cartoon}, we need to iteratively multiply $\hat{\mathcal{M}}_i^\Sigma$ with the cartesian product of reachable set $\hat{\mathcal{R}}_k$ and the input set $\mathcal{U}_k$ to obtain the next reachable set $\hat{\mathcal{R}}_{k+1}$.
Therefore, we provide an exact multiplication method for CPMZ with CPZ. 

Consider ${\mathcal{Y}} = \zono{C_{\mathcal{Y}}, G_{\mathcal{Y}}, E_{\mathcal{Y}}, A_{\mathcal{Y}}, B_{\mathcal{Y}}, R_{\mathcal{Y}}, \id_{\mathcal{Y}}}_\text{CPMZ}\subset \R^{n_x \times n }$ and ${\mathcal{P}} = \zono{c_{\mathcal{P}}, G_{\mathcal{P}}, E_{\mathcal{P}}, A_{\mathcal{P}}, b_{\mathcal{P}}, R_{\mathcal{P}}, \id_{\mathcal{P}}}_\text{CPZ}\subset \R^{n}$. Since both sets may contain shared dependent factors, consistent indexing of exponent matrices is required. To address this, we leverage the \operator{mergeID} operator, which constructs a unified exponent representation to preserve these dependencies:
\begin{align}
  \operator{mergeID}(\mathcal{Y},\mathcal{P}) = &\big \{ \underbrace{\langle C_{\mathcal{Y}}, G_{\mathcal{Y}}, \overline{E}_\mathcal{Y}, A_{\mathcal{Y}}, B_{\mathcal{Y}}, \overline{R}_\mathcal{Y}, \id_\mathcal{YP} \rangle_\text{CPMZ}}_{\bar{\mathcal{Y}}}, \nonumber \\& 
  \underbrace{\langle c_{\mathcal{P}}, G_{\mathcal{P}}, \overline{E}_{\mathcal{P}}, A_{\mathcal{P}}, b_{\mathcal{P}}, \overline{R}_{\mathcal{P}}, \id_{\mathcal{YP}} \rangle_\text{CPZ}}_{\bar{\mathcal{P}}} \big \}.  \label{mergetwoexact}
\end{align}

\begin{proposition}(Exact Multiplication between CPMZ and CPZ) \label{prop:multi}  
Given a CPMZ ${\mathcal{Y}} = \zono{C_{\mathcal{Y}}, G_{\mathcal{Y}}, E_{\mathcal{Y}}, A_{\mathcal{Y}}, B_{\mathcal{Y}}, R_{\mathcal{Y}}, \id_{\mathcal{Y}}}_\text{CPMZ}\subset \R^{n_x \times n }$ and a CPZ $\mathcal{P} = \langle c_{\mathcal{P}}, G_{\mathcal{P}}, {E}_{\mathcal{P}}, A_{\mathcal{P}}, b_{\mathcal{P}}, {R}_{\mathcal{P}}, \id_{\mathcal{P}} \rangle_\text{CPZ}\subset \R^{n}$, the following identity holds
\begin{multline}
\mathcal{Y} \otimes \mathcal{P} =  \bigzono{C_{\mathcal{Y}} c_{\mathcal{P}},\begin{bmatrix} G_{\mathcal{Y}}c_{\mathcal{P}} & C_{\mathcal{Y}} G_{\mathcal{P}}& G_{f} \end{bmatrix} , E_{\mathcal{YP}}, \\
    A_{\mathcal{Y}\mathcal{P}}, {B}_{\mathcal{Y}\mathcal{P}} , R_{\mathcal{YP}} ,\id_{\mathcal{YP}}}_\text{CPZ},  \label{eq:matczono}
\end{multline}
where $\mathcal{Y}\otimes \mathcal{P} \subset \R^{n_x}$ and
\begin{align*}
    A_{\mathcal{Y}\mathcal{P}} &= \begin{bmatrix} \text{vec}(A^{(1)}_{\mathcal{Y}}) & \dots & \text{vec}
    (A_{\mathcal{Y}}^{(\gamma)}) & 0_{n_{c}n_{a}  \times q_{\mathcal{P}}} \\
     0_{m_{\mathcal{P}}  \times 1} & \dots & 0_{m_{\mathcal{P}}  \times 1} & A_{\mathcal{P}}
    \end{bmatrix} \\
    {B}_{\mathcal{Y}\mathcal{P}} &=\begin{bmatrix} \text{vec}(B_{\mathcal{Y}}) \\  b_{\mathcal{P}}\end{bmatrix} \\
    E_\mathcal{YP} &= \bigg[ \overline{E}_{\mathcal{Y}}, \overline{E}_{\mathcal{P}}, 
    \Big[\overline{E}_{\mathcal{Y}}^{(\cdot,1)} + \overline{E}_{\mathcal{P}}^{(\cdot,1)}\Big],
    \dots,
    \Big[\overline{E}_{\mathcal{Y}}^{(\cdot,1)} + \overline{E}_{\mathcal{P}}^{(\cdot,h_\mathcal{P})}\Big], \\
    & \qquad\dots, \Big[\overline{E}_{\mathcal{Y}}^{(\cdot,\gamma)} + \overline{E}_{\mathcal{P}}^{(\cdot,1)}\Big],
    \dots, \Big[\overline{E}_{\mathcal{Y}}^{(\cdot,\gamma)}+\overline{E}_{\mathcal{P}}^{(\cdot,h_\mathcal{P})}\Big] \bigg]  
    \\
    {R_\mathcal{YP}} &= \begin{bmatrix} \overline{R}_{\mathcal{Y}} & \overline{R}_{\mathcal{P}} \end{bmatrix} \\
    G_{f} &= \begin{bmatrix} g_{f}^{(1)} & \dots & g_{f}^{(h_\mathcal{P}\gamma)} \end{bmatrix}
\end{align*}
with $\overline{E}_\mathcal{Y}\in\R^{a \times \gamma}$, $\overline{E}_\mathcal{P}\in\R^{a \times h_{\mathcal{P}}}$, $\overline{R}_\mathcal{Y}\in\R^{a \times \gamma}$, $\overline{R}_\mathcal{P}\in\R^{a \times h_{\mathcal{P}}}$, and $\id_{\mathcal{YP}}\in\N^{1 \times a }$ are induced by the 
$\operatorname{mergeID}(\mathcal{Y},\mathcal{P})$ operation
as defined in Proposition~\ref{prop:mergeID}. And
\begin{displaymath}
    g_{f}^{(k)} =  G^{(i)}_{\mathcal{Y}} G_{\mathcal{P}}^{(\cdot,j)}, \quad k=h_{\mathcal{P}} (i-1) + j,
\end{displaymath}
for $i=1,\dots,\gamma$,\;$j=1,\dots,h_\mathcal{P}$.
\end{proposition}
\begin{proof}
    Let $\hat{\mathcal{P}}$ be the right-hand side of \eqref{eq:matczono} and let $Y\in \mathcal{Y}$ and $p \in \mathcal{P}$. 
    We will prove that $\mathcal{Y}\mathcal{P} \subseteq \hat{\mathcal{P}}$ and $\hat{\mathcal{P}} \subseteq \mathcal{Y}\mathcal{P}$, for all $Y\in\mathcal{Y}$ and $p\in\mathcal{P}$. 
    Note that with the implementation of $\operator{mergeID}(\mathcal{Y},\mathcal{P})$, $Y$ and $p$ can be written as:
    \begin{align}
    \exists \hat{\alpha}_{([1:a])} &: & Y &= C_{\mathcal{Y}} + \sum_{i=1}^{\gamma}\bigg( \prod_{k=1}^{a}
    \hat{\alpha}_{(k)}^{\overline{E}_\mathcal{Y}^{(k,i)}} \bigg)  G_{\mathcal{Y}}^{(i)}
    \nonumber\\
    \exists \hat{\alpha}_{([1:a])} &: & p 
    &=  c_{\mathcal{P}} +   \sum_{i=1}^{h_\mathcal{P}} \bigg( \prod_{k=1}^{a}
    \hat{\alpha}_{(k)}^{\overline{E}_\mathcal{P}^{(k,i)}} \bigg) G_{\mathcal{P}}^{(\cdot,i)} \label{alphadetail}
    \end{align}
    where $a = |\mathcal{H}|+\gamma$, $\mathcal{H} = \left\{ i~ |~ \id_{{\mathcal{P}}(i)} \not\in \id_{\mathcal{Y}} \right\}$. 
    Thus, we have
    \begin{multline}
    Yp = C_{\mathcal{Y}} c_{\mathcal{P}} +  \sum_{i=1}^{\gamma}\bigg( \prod_{k=1}^{a}
    \hat{\alpha}_{(k)}^{\overline{E}_\mathcal{Y}^{(k,i)}} \bigg)  G_{\mathcal{Y}}^{(i)}  c_{\mathcal{P}} \\
    + C_{\mathcal{Y}}\sum_{i=1}^{h_{\mathcal{P}}} \bigg( \prod_{k=1}^{a}
    \hat{\alpha}_{(k)}^{\overline{E}_\mathcal{P}^{(k,i)}} \bigg)  G_{\mathcal{P}}^{(\cdot,i)} \\
    + \sum_{i=1}^{\gamma} \sum_{j=1}^{h_{\mathcal{P}}} \bigg( \prod_{k=1}^{a}
    \hat{\alpha}_{(k)}^{\overline{E}_\mathcal{Y}^{(k,i)}} \bigg) \bigg( \prod_{k=1}^{a}
    \hat{\alpha}_{(k)}^{\overline{E}_\mathcal{P}^{(k,j)}} \bigg)  G_{\mathcal{Y}}^{(i)}  G_{\mathcal{P}}^{(\cdot,j)} . \label{eq:P_details}
    \end{multline}
    For the second and third terms on the right-hand side of \eqref{eq:P_details}, we have two sets of factors, each containing $\gamma$ and $h_\mathcal{P}$ elements. Consequently, the first $\gamma + h_\mathcal{P}$ entries of $\hat{\alpha}$ and the columns of $E_\mathcal{YP}$ are defined accordingly as follows:
    \begin{align}
    \hat{\alpha}_{([1:\gamma+h_\mathcal{P}])} &=\bigg[\prod_{k=1}^{a}
    \hat{\alpha}_{(k)}^{\overline{E}_\mathcal{Y}^{(k,1)}}, \dots,\prod_{k=1}^{a}
    \hat{\alpha}_{(k)}^{\overline{E}_\mathcal{Y}^{(k,\gamma)}} \nonumber\\
    & \hspace{1cm} 
    \prod_{k=1}^{a} \hat{\alpha}_{(k)}^{\overline{E}_\mathcal{P}^{(k,1)}}, \dots, 
    \prod_{k=1}^{a} \hat{\alpha}_{(k)}^{\overline{E}_\mathcal{P}^{(k,h_\mathcal{P})}}\bigg] \\{E}_{\mathcal{YP}}^{([1:\gamma+h_\mathcal{P}])} &= \Big[ \overline{E}_{\mathcal{Y}},\overline{E}_{\mathcal{P}}\Big].
\end{align} 
Because of $\operator{mergeID}(\mathcal{Y},\mathcal{P})$ in \eqref{mergetwoexact}, the fourth term on the right-hand side of \eqref{eq:P_details} can be expressed as:
\begin{multline}
\sum_{i=1}^{\gamma} \sum_{j=1}^{h_{\mathcal{P}}} \bigg( \prod_{k=1}^{a}
  \hat{\alpha}_{(k)}^{\overline{E}_\mathcal{Y}^{(k,i)}}  
  \hat{\alpha}_{(k)}^{\overline{E}_\mathcal{P}^{(k,j)}} \bigg)  G_{\mathcal{Y}}^{(i)}  G_{\mathcal{P}}^{(\cdot,j)}= \\
  \sum_{i=1}^{\gamma} \sum_{j=1}^{h_{\mathcal{P}}} \bigg( \prod_{k=1}^{a}
  \hat{\alpha}_{(k)}^{\overline{E}_\mathcal{Y}^{(k,i)}+  
  \overline{E}_\mathcal{P}^{(k,j)}} \bigg)  G_{\mathcal{Y}}^{(i)}  G_{\mathcal{P}}^{(\cdot,j)}. \label{1:n+p}
\end{multline} 
Concatenating the factors in \eqref{1:n+p}, we have
\begin{multline*}
\hat{\alpha}_{([\gamma+h_\mathcal{P}+1:\gamma+h_\mathcal{P}+\gamma h_\mathcal{P}])} \\ 
=\bigg[\prod_{k=1}^{a}
  \hat{\alpha}_{(k)}^{\overline{E}_\mathcal{Y}^{(k,1)}+  
  \overline{E}_\mathcal{P}^{(k,1)}}, \dots, \prod_{k=1}^{a}
  \hat{\alpha}_{(k)}^{\overline{E}_\mathcal{Y}^{(k,1)}+  
  \overline{E}_\mathcal{P}^{(k,h_\mathcal{P})}}, \dots, \\
  \prod_{k=1}^{a}
  \hat{\alpha}_{(k)}^{\overline{E}_\mathcal{Y}^{(k,\gamma)}+  \overline{E}_\mathcal{P}^{(k,1)}}, \dots,\prod_{k=1}^{a}
  \hat{\alpha}_{(k)}^{\overline{E}_\mathcal{Y}^{(k,\gamma)}+  \overline{E}_\mathcal{P}^{(k,h_\mathcal{P})}}\bigg],
\end{multline*}
 which results in $E_\mathcal{YP}^{([\gamma+h_\mathcal{P}+1:\gamma+h_\mathcal{P}+\gamma h_\mathcal{P}])}$ and $G_{f}$ as follows:
 \begin{align*}
    &E_\mathcal{YP}^{([\gamma+h_\mathcal{P} + 1 : \gamma+h_\mathcal{P} + \gamma h_\mathcal{P}])}= \bigg[  
    \Big[\overline{E}_{\mathcal{Y}}^{(\cdot,1)} + \overline{E}_{\mathcal{P}}^{(\cdot,1)}\Big],
    \dots, 
    \Big[\overline{E}_{\mathcal{Y}}^{(\cdot,1)} \\
    & + \overline{E}_{\mathcal{P}}^{(\cdot,h_\mathcal{P})}\Big], 
    \dots, \Big[\overline{E}_{\mathcal{Y}}^{(\cdot,\gamma)} + \overline{E}_{\mathcal{P}}^{(\cdot,1)}\Big],\dots, \Big[\overline{E}_{\mathcal{Y}}^{(\cdot,\gamma)}+\overline{E}_{\mathcal{P}}^{(\cdot,h_\mathcal{P})}\Big] \bigg] \\
    &G_{f} = \bigg[ G^{(1)}_{\mathcal{Y}} G_{\mathcal{P}}^{(\cdot,1)},{...}, G^{(1)}_{\mathcal{Y}} G_{\mathcal{P}}^{(\cdot,h_\mathcal{P})},{...},G^{(\gamma)}_{\mathcal{Y}} G_{\mathcal{P}}^{(\cdot,h_\mathcal{P})} \bigg].
    \end{align*}
    Secondly, we find the constraints on $\hat{\alpha}_{([1:a])}$ in \eqref{alphadetail}. For  $Y\in \mathcal{Y}$ and  $p \in \mathcal{P}$, $Yp$ defined in~\eqref{eq:P_details} should satisfy the constraints simultaneously.
    From \eqref{mergetwoexact}, we have
    \begin{multline}
    \sum_{i=1}^{\gamma} \bigg( \prod_{k=1}^{\gamma}
(\alpha_{\mathcal{Y}}^{(k)})^{R_\mathcal{Y}^{(k,i)}} \bigg) \text{vec}(A_{\mathcal{Y}}^{(i)}) =\\ 
    \sum_{i=1}^{\gamma} \bigg( \prod_{k=1}^{a}
    \hat{\alpha}_{(k)}^{\overline{R}_\mathcal{Y}^{(k,i)}} \bigg) \text{vec}(A_{\mathcal{Y}}^{(i)})=
    \text{vec}(B_{\mathcal{Y}})
    \label{constraint1}
    \end{multline} 
    and
    \begin{multline}
    \sum_{i=1}^{q_\mathcal{P}} \bigg( \prod_{k=1}^{h_\mathcal{P}}
(\alpha_{\mathcal{P}}^{(k)})^{R_\mathcal{P}^{(k,i)}} \bigg)A_{\mathcal{P}}^{{(\cdot,i)}} = \\
    \sum_{i=1}^{q_\mathcal{P}} \bigg( \prod_{k=1}^{a}
    \hat{\alpha}_{(k)}^{\overline{R}_\mathcal{P}^{(k,i)}} \bigg)A_{\mathcal{P}}^{{(\cdot,i)}}=b_{\mathcal{P}}.
    \label{constraint2}
    \end{multline}
    Combining \eqref{constraint1} and \eqref{constraint2}, the following holds
    \begin{multline}
    \sum_{i=1}^{\gamma} \bigg( \prod_{k=1}^{a}
    \hat{\alpha}_{(k)}^{\overline{R}_\mathcal{Y}^{(k,i)}} \bigg)\begin{bmatrix} \text{vec}(A_{\mathcal{Y}}^{{(\cdot,i)}}) \\  0_{m_\mathcal{P}\times1}\end{bmatrix}\\
    +\sum_{i=1}^{q_\mathcal{P}} \bigg( \prod_{k=1}^{a}
    \hat{\alpha}_{(k)}^{\overline{R}_\mathcal{P}^{(k,i)}} \bigg)\begin{bmatrix} 0_{n_{c}n_{a}\times1} \\  A_{\mathcal{P}}^{{(\cdot,i)}}\end{bmatrix}=\begin{bmatrix} \text{vec}(B_{\mathcal{Y}}) \\  b_{\mathcal{P}}\end{bmatrix}
    \label{constraint3} 
    \end{multline}
    which results in $A_{\mathcal{Y}\mathcal{P}}$ and $B_{\mathcal{Y}\mathcal{P}}$. 
    Thus, $Yp \in \hat{\mathcal{P}}$ and therefore $\mathcal{Y} \mathcal{P} \subseteq \hat{\mathcal{P}}$. 
    Conversely, let $\hat{p} \in \hat{\mathcal{P}}$, then 
    \begin{align*}
    \exists \hat{\alpha}_{([1:{\gamma+h_\mathcal{P}+\gamma h_\mathcal{P}}])}&:  \hat p =\nonumber\hat c + \sum_{i=1}^{\gamma+h_\mathcal{P}+\gamma h_\mathcal{P}}\bigg( \prod_{k=1}^{a}
    \hat{\alpha}_{(k)}^{{E}_\mathcal{YP}^{(k,i)}} \bigg)  G_{f}^{(i)}.
    \end{align*}
    By partitioning
    \begin{multline*}
    \hat{\alpha}_{([1:{\gamma+h_\mathcal{P}+\gamma h_\mathcal{P}}])}=\bigg[\prod_{k=1}^{a}
    \hat{\alpha}_{(k)}^{\overline{E}_\mathcal{Y}^{(k,[1:\gamma])}},\prod_{k=1}^{a}
    \hat{\alpha}_{(k)}^{\overline{E}_\mathcal{P}^{(k,[1:h_\mathcal{P}])}},\\
    \prod_{k=1}^{a}
    \hat{\alpha}_{(k)}^{\overline{E}_\mathcal{Y}^{(k,1)}+  
    \overline{E}_\mathcal{P}^{(k,1)}}, \dots,\prod_{k=1}^{a}
    \hat{\alpha}_{(k)}^{\overline{E}_\mathcal{Y}^{(k,1)}+  
    \overline{E}_\mathcal{P}^{(k,h_\mathcal{P})}},\dots,\\
    \prod_{k=1}^{a}
    \hat{\alpha}_{(k)}^{\overline{E}_\mathcal{Y}^{(k,\gamma)}+  \overline{E}_\mathcal{P}^{(k,1)}}, \dots,\prod_{k=1}^{a}
    \hat{\alpha}_{(k)}^{\overline{E}_\mathcal{Y}^{(k,\gamma)}+  \overline{E}_\mathcal{P}^{(k,h_\mathcal{P})}}\bigg]  
    \label{eq:beta_c1_3}
    \end{multline*}
    it follows that there exist $Y\in \mathcal{Y}$ and $p \in \mathcal{P}$ such that $\hat{p} = Yp$.
    Meanwhile, since $\hat{p}\in\hat{\mathcal{P}}$, it holds that
    \begin{align}
    \sum_{i=1}^{q_\mathcal{P}+\gamma} \bigg( \prod_{k=1}^{a}
    \hat{\alpha}_{(k)}^{{R}_\mathcal{YP}^{(k,i)}} \bigg)A_{\mathcal{YP}}^{{(\cdot,i)}}=\begin{bmatrix} \text{vec}(B_{\mathcal{Y}}) \\  b_{\mathcal{P}}\end{bmatrix}
    \end{align}
    which satisfies the constraints in \eqref{constraint3}. Therefore, $\hat{p} \in \mathcal{Y} \mathcal{P}$ and thus $ \hat{\mathcal{P}} \subseteq \mathcal{Y}   \mathcal{P}$.
\end{proof}

As noted before, $\mathcal{N}=\zono{C_{\mathcal{N}}, G_{\mathcal{N}}, A_{\mathcal{N}}, B_{\mathcal{N}}, \id_{\mathcal{N}}}_\text{CMZ}$ is a special case of CPMZ and thus can be reformulated as $\mathcal{Y}=\zono{C_{\mathcal{N}}, G_{\mathcal{N}}, E_{\mathcal{N}}, A_{\mathcal{N}}, B_{\mathcal{N}}, R_{\mathcal{N}}, \id_{\mathcal{N}}}_\text{CPMZ}$ with $E_{\mathcal{N}} = R_{\mathcal{N}} = I_{\gamma_{\mathcal{N}}}$. 
Then, according to Proposition~\ref{prop:multi}, the CPMZ is multiplied exactly with a CPZ to yield a new CPZ, as illustrated by Fig.~\ref{fig:proof-exact-mult-cartoon}.

\begin{figure}[!h]
    \centering
    \includegraphics[width=0.8\linewidth]{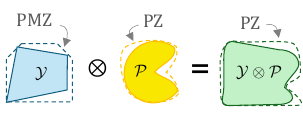}
    \caption{Illustration of the proof of Proposition \ref{prop:multi}, where the dashed lines indicate the boundary of PMZ and PZs, respectively. Let $\textrm{c}(\mathcal{Y})$ and $\textrm{c}(\mathcal{P})$ denote the constraints associated with CPMZ $\mathcal{Y}$ and CPZ $\mathcal{P}$, respectively. The exact multiplication of $\mathcal{Y}:=\{\textrm{PMZ}:\textrm{c}(\mathcal{Y}) \textrm{ hold}\}$ with $\mathcal{P}:=\{\textrm{PZ}:\textrm{c}(\mathcal{P}) \textrm{ hold}\}$ is a CPZ $\mathcal{Y}\otimes \mathcal{P}=\{\textrm{PZ}:\textrm{c}(\mathcal{Y}\otimes \mathcal{P}) \textrm{ hold}\}$, with constraints $\textrm{c}(\mathcal{Y}\otimes \mathcal{P}):= \textrm{Proj}_{\R^{n}} c(\mathcal{Y})$ and $c(\mathcal{P})$, where $\textrm{Proj}_{\R^{n}} c(\mathcal{Y})$ is the projection of $c(\mathcal{Y})$ from $\R^{n_x\times n}$ to $\R^{n}$.}
    \label{fig:proof-exact-mult-cartoon}
\end{figure}

\Comp The quantities \(\mathrm{vec}(B_{\mathcal{Y}})\), \(A_{\mathcal{Y}\mathcal{P}}\), \(E_{\mathcal{Y}\mathcal{P}}\), \(C_{\mathcal{Y}}c_{\mathcal{P}}\), \(G_{\mathcal{Y}}c_{\mathcal{P}}\), \(C_{\mathcal{Y}}G_{\mathcal{P}}\), and \(G_f\) are obtained using standard block concatenation, vectorization, and matrix multiplications, including matrix--vector multiplications and matrix--matrix multiplications. In particular, \(\mathrm{vec}(B_{\mathcal{Y}})\) has complexity \(\mathcal{O}(n_c n_a)\), \(A_{\mathcal{Y}\mathcal{P}}\) has complexity \(\mathcal{O}(n_c n_a \gamma)\), \(E_{\mathcal{Y}\mathcal{P}}\) has complexity \(\mathcal{O}((|\mathcal{H}|+\gamma)\gamma h_{\mathcal{P}})\), \(C_{\mathcal{Y}}c_{\mathcal{P}}\), \(G_{\mathcal{Y}}c_{\mathcal{P}}\), and \(C_{\mathcal{Y}}G_{\mathcal{P}}\) have complexities \(\mathcal{O}(n_x n)\), \(\mathcal{O}(\gamma n_x n)\), and \(\mathcal{O}(n_x n h_{\mathcal{P}})\), respectively, and \(G_f\) has complexity \(\mathcal{O}(\gamma n_x n h_{\mathcal{P}})\). Therefore, the exact multiplication in \(\eqref{eq:matczono}\) has overall complexity \(\mathcal{O}\!\bigl(n_c n_a \gamma + (|\mathcal{H}| + \gamma)\gamma h_{\mathcal{P}} + \gamma n_x n h_{\mathcal{P}}\bigr)\). \hfill \(\diamond\)

\begin{remark}
Proposition~\ref{prop:multi} provides an exact set operation in the sense that the resulting set is represented without introducing additional conservatism. However, repeated application of this exact multiplication typically increases the representation size of the resulting CPMZ and CPZ objects, which can lead to a rapidly growing computational burden. In the numerical experiments, we therefore apply a representation-size reduction step for CPZs, following the reduction procedures in \cite{kochdumper2023constrained}. This reduction step is the only source of over-approximation in our implementation; all other operations are carried out exactly as stated in Proposition~\ref{prop:multi}.
\end{remark}

\begin{remark}
The complexity estimates above are per propagation step. In multi-step reachability, repeated set operations may increase the number of generators, constraints, and dependent identifiers, which in turn increases the per-step computational cost. This growth can be controlled without introducing conservatism by applying exact simplification procedures that remove redundant constraints and eliminate inactive variables in the constraint description, while preserving the represented set; see, e.g., \cite{scott2016constrained,marechal2017efficient}.
\end{remark}

To preserve dependencies between generators during the addition of two CPZs and prevent over-approximation, we adopt the exact addition operation proposed by \cite{kochdumper2020sparse}.
By retaining dependency information, the resulting sets are less conservative, improving the computational efficiency and reliability of the resulting CPZ.

\begin{proposition}(Exact Addition between CPZs)
Let $\mathcal{P}_1 = \zono{c_{1},G_{1},E_{1},A_{1},b_{1},R_{1},\id_{1}}_\text{CPZ} \subset \R^{n}$ and $\mathcal{P}_2 = \zono{c_{2},G_{2},E_{2},A_{2},b_{2},R_{2},\id_{2}}_\text{CPZ} \subset \R^{n}$. Then, the exact addition is given by 
\begin{align}
\mathcal{P}_1 \boxplus \mathcal{P}_2
&= \bigzono{
\left[\begin{array}{l}
c_{1} \\
c_{2}
\end{array}\right],
\left[\begin{array}{cc}
G_{1} & 0_{n \times h_{2}} \\
0_{n \times h_{1}} & G_{2}
\end{array}\right],\notag\\
&\qquad
\left[\begin{array}{cc}
\overline{E}_{1} & \overline{E}_{2}
\end{array}\right],
\left[\begin{array}{cc}
A_{1} & 0_{n_{c_{1}} \times q_{2}} \\
0_{n_{c_{2}} \times q_{1}} & A_{2}
\end{array}\right],\notag\\
&\qquad
\left[\begin{array}{l}
b_{1} \\
b_{2}
\end{array}\right],
\left[\begin{array}{cc}
\overline{R}_{1} & \overline{R}_{2}
\end{array}\right],
\overline \id
}_\text{CPZ}.
\end{align}
where $\overline{E}_{1},\overline{E}_{2},
\overline{R}_{1},\overline{R}_{2}$ and 
$\overline {\id}$ are induced by the 
$\operatorname{mergeID}(\mathcal{P}_1,\mathcal{P}_2)$ operation
as defined in Proposition~\ref{prop:mergeID}.
   \label{prop:exactAddition}
\end{proposition}

\begin{algorithm}[h]
\caption{Data-driven Non-Convex Reachability Analysis for LTI system}
\label{alg:offline-online}
\textbf{Input:} 
Initial input-state data $D_0 = (X_0^+, X_0^-, U_0^-)$, initial set $\mathcal{X}_0$, noise zonotope $\mathcal{Z}_w$, input zonotopes $\mathcal{U}_k$, batch length \(L\ge n_x+n_u\); \\
\textbf{Output:} 
Reachable sets $\hat{\mathcal{R}}_k$ for $k \in \N$.

\noindent\textbf{Offline: Initialization}
\begin{algorithmic}[1]
  \State $\hat{\mathcal{R}}_0 \gets \mathcal{X}_0$
  \State $\mathcal{M}_0^\Sigma \gets ( X_0^+ - \mathcal{M}_w ) \begin{bmatrix} X_0^- \\ U_0^- \end{bmatrix}^\dagger$
 \State Set $\hat{\mathcal{M}}^\Sigma \gets \mathcal{M}_0^\Sigma$
\State Set $X^+ \gets [~]$, $X^- \gets [~]$, $U^- \gets [~]$
\end{algorithmic}

\noindent\textbf{Online: Set Refinement and Reachability Analysis }
\begin{algorithmic}[1]

\While{$k\geq 1$}
    \State Construct online data matrices $X^+ = [X^+ ~ x_{(k)}]$, $X^- = [X^- ~ x_{(k-1)}]$, $U^- =[U^- ~ u_{(k-1)}]$
    
    \If{\(\mathrm{cols}(X^-)=L\) \textbf{and} $\textrm{rank}\left(\begin{bmatrix}X^- \\ U^-\end{bmatrix}\right)=n_x + n_u$}  
    
        \State $\mathcal{M}^\Sigma \gets ( X^+ - \mathcal{M}_w ) \begin{bmatrix} X^- \\ U^- \end{bmatrix}^\dagger$
    \State $\hat{\mathcal{M}}^\Sigma \gets \mathcal{M}^\Sigma \cap \hat{\mathcal{M}}^\Sigma$
    \State $X^+ \gets [~]$, $X^- \gets [~]$, $U^- \gets [~]$
    
\EndIf
\State $\hat{\mathcal{R}}_{k+1} \gets \hat{\mathcal{M}}^\Sigma\otimes ( \hat{\mathcal{R}}_k \times \mathcal{U}_k) \boxplus \mathcal{Z}_w$

\State $k \gets k+1$.
\EndWhile
  
\end{algorithmic}
\end{algorithm}

The proposed data-driven reachability analysis procedure for LTI systems is summarized in Algorithm~\ref{alg:offline-online}.
The algorithm consists of an offline initialization phase followed by an online
set refinement and reachable set propagation phase.
During initialization, an initial set of admissible system matrices
$\hat{\mathcal{M}}^\Sigma$ is computed from the available input--state data and
the reachable set is initialized to the given initial set.

As new input--state measurements become available online, the set of admissible
system matrices is recursively refined through set intersection whenever the
rank condition is satisfied.
At each time step, the reachable set is propagated forward by applying the
refined matrix zonotope $\hat{\mathcal{M}}^\Sigma$ to the Cartesian product of
the current reachable set and the input set, followed by the incorporation of
the process noise zonotope.

The use of CPMZs in combination with
CPZs enables an exact set-based propagation of
reachable sets for LTI systems.
In contrast to earlier approaches relying solely on CMZ--CZ multiplication,
the proposed formulation exploits the richer algebraic structure of CPMZs,
leading to a less conservative and more expressive reachability analysis
framework.

\begin{remark}
\label{rem:cmz_sideinfo_lti}
In Algorithm~\ref{alg:offline-online}, the model set construction in Line~2 and the subsequent refinement via intersection in Line~4 can leverage enhanced CMZ descriptions as proposed in~\cite{Alanwar2023Datadriven}. In particular, the CMZ characterization can be tightened by using an improved noise MZ description (e.g., $\mathcal{N}_w$), and by incorporating additional \emph{side information} on the true system matrices, formulated as the element-wise inequality $\big|\bar{Q}\begin{bmatrix}A_{\mathrm{tr}} & B_{\mathrm{tr}}\end{bmatrix}-\bar{Y}\big|\le \bar{R}$, where $\bar{Q}\in\mathbb{R}^{n_s\times n_x}$ and $\bar{Y},\bar{R}\in\mathbb{R}^{n_s\times (n_x+n_u)}$ encode known structural constraints (e.g., bounds, sparsity, or sign patterns) on $\begin{bmatrix}A_{\mathrm{tr}} & B_{\mathrm{tr}}\end{bmatrix}$. These refinements yield a tighter data-consistent model set and can thereby reduce conservatism in the resulting reachable set over-approximations. Since CMZs are a special case of CPMZs, such enhancements can be incorporated seamlessly without modifying the subsequent reachability propagation steps of our framework.
\end{remark}

\section{Model Based and Data-Driven Reachability for Polynomial Systems}\label{sec:reachnonlinear}

This section extends the proposed algebraic framework to polynomial dynamical systems.
The focus is on exact set-based propagation of reachable sets under polynomial dynamics. The notion of exactness adopted in the model-based polynomial reachability framework follows the same algebraic principle introduced in Section~\ref{sec:preliminaries}, thereby allowing polynomial dynamics to be evaluated at the set level using exactly the same set operations as in the data-driven case.

\begin{remark}[Exactness for polynomial dynamics]
In the polynomial case, exactness is understood in the algebraic sense defined in Section~\ref{sec:preliminaries}, and refers to the direct evaluation of monomials at the set level using CPZ arithmetic.
Accordingly, no interval arithmetic, bounding, or relaxation of polynomial terms is introduced during propagation.
\end{remark}

We consider a discrete-time nonlinear control system described by a polynomial mapping
\begin{align}
    x_{(k+1)} = f_p(x_{(k)},u_{(k)}) + w_{(k)},
    \label{eq:sysnonlin_poly}
\end{align}
where the function $f_p:\mathbb{R}^{n_x}\times\mathbb{R}^{n_u}\rightarrow\mathbb{R}^{n_x}$ is assumed to be polynomial.
For notational convenience, the explicit dependence of the signals on the time index $k$ is omitted whenever it is clear from the context.

Define $n_z = n_x + n_u$ and introduce the stacked vector
\begin{align}
    z=\begin{bmatrix} x^T & u^T \end{bmatrix}^T
    = \begin{bmatrix} z_1^T & \dots & z_{n_z}^T \end{bmatrix}^T
    \in \mathbb{R}^{n_z},
\end{align}
where, with a slight abuse of notation, the components of $z$ collect both the state and input variables.

By a polynomial system, we mean that $f_p(z)\in\mathbb{R}[z]^{n_x}$, i.e., each component of $f_p(z)$ belongs to the space of multivariate polynomials in the variables $z_1,\dots,z_{n_z}$.
In particular, the $i$-th component of $f_p(z)$ admits the representation
\begin{align*}
    f_p^{(i)}(z)
    = \sum_{j=1}^{m_i} \theta_j
    z_1^{\alpha_{j,1}} z_2^{\alpha_{j,2}} \dots z_{n_z}^{\alpha_{j,n_z}},
\end{align*}
where $m_i$ denotes the number of monomial terms in $f_p^{(i)}(z)$, $\theta_j\in\mathbb{R}$ are scalar coefficients, and the exponent vectors
\[
\alpha_j = \begin{bmatrix} \alpha_{j,1} & \dots & \alpha_{j,n_z} \end{bmatrix}^T \in \mathbb{N}_0^{n_z}
\]
satisfy the degree constraint $\sum_{i=1}^{n_z}\alpha_{j,i}\le d$ for all $j\in\{1,\dots,m_i\}$, with $d>0$ denoting the maximum polynomial degree.

Following \cite{conf:polyDissipat}, the polynomial mapping $f_p(z)$ can be written in the compact form
\begin{align}
    f_p(z) = \Theta_{\text{tr}}\, h(z),
    \label{eq:pcg}
\end{align}
where $h(z)\in\mathbb{R}[z]^{m_a}$ is a vector collecting a set of monomials that includes all monomials appearing in $f_p(z)$, and $\Theta_{\text{tr}}\in\mathbb{R}^{n_x\times m_a}$ contains the corresponding (unknown) coefficients. 
The monomial basis $h(z)$ can be constructed, for example, based on a known upper bound on the polynomial degree of $f_p(z)$.
If additional structural information on $f_p(z)$ is available, $h(z)$ may be chosen to consist exactly of the monomials present in the polynomial mapping.

\begin{remark}
The results are presented for polynomial systems using the monomial basis in \eqref{eq:pcg}. The analysis primarily exploits the fact that the unknown coefficients enter linearly through \(\Theta_{\mathrm{tr}}\). Accordingly, the same framework extends directly to linearly parameterized models \(f(z)=\Theta h(z)\) with other choices of known basis functions \(h(\cdot)\), provided that the corresponding set operations can be carried out in the adopted set representation. Moreover, certain nonlinear parameterizations can be handled when they admit an equivalent linear-in-coefficients representation via a suitable reparameterization.
\end{remark}

\begin{remark}
The basis dimension \(m_a\) depends on \(n_z=n_x+n_u\) and the degree bound \(d\). If \(h(z)\) contains all monomials of total degree at most \(d\), then \(m_a=\binom{n_z+d}{d}\). When additional structure is available, \(h(z)\) can be chosen as a reduced monomial set that matches the assumed support of \(f_p(z)\), which lowers \(m_a\) and the associated computational cost.
\end{remark}

In the following, we first present an exact model-based reachable set computation method for polynomial systems.
We then introduce an exact data-driven reachable set computation framework, which computes reachable sets using only data and mild structural assumptions.

\subsection{Model-based reachability analysis}\label{sec:poly}
\begin{proposition}
	(Exact Cartesian Product) Given $\mathcal{P}_1 = \langle c_{1}, G_{1}, E_{1}, A_{1},b_{1}, \CE_{1}, {\id}_{1} \rangle_{CPZ} \linebreak[3] \subset \Rn$ and $\mathcal{P}_2 = \langle c_{2}, G_{1}, E_{1}, A_{1}, b_{1}, \CE_{1} ,{\id}_{2}\rangle_{CPZ} \subset \R^w$, their exact Cartesian product is
	\begin{align}
			\mathcal{P}_1 \boxtimes \mathcal{P}_2 =& \bigg \langle \begin{bmatrix} c_{1} \\ c_{2} \end{bmatrix} \begin{bmatrix} G_{1} & \mathbf{0} \\ \mathbf{0} & G_{1} \end{bmatrix}, \begin{bmatrix} \bar{E}_{1} & \mathbf{0} \\ \mathbf{0} & \bar{E}_{1} \end{bmatrix},\nonumber\\& \quad \begin{bmatrix} A_{1} & \mathbf{0} \\ \mathbf{0} & A_{1} \end{bmatrix},\begin{bmatrix} b_{1} \\ b_{1} \end{bmatrix}, \begin{bmatrix} \bar{\CE}_{1} & \mathbf{0} \\ \mathbf{0} & \bar{\CE}_{1} \end{bmatrix}, \overline {\id}  \bigg \rangle_{\text{CPZ}},
	\end{align}	\label{prop:cartProduct}
where $\overline{E}_{1},\overline{E}_{2},
\overline{R}_{1},\overline{R}_{2}$ and 
$\overline {\id}$ are induced by the 
$\operatorname{mergeID}(\mathcal{P}_1,\mathcal{P}_2)$ operation
as defined in Proposition~\ref{prop:mergeID}.
\end{proposition}
\begin{proof}
By Proposition~\ref{prop:mergeID}, the operator $\operatorname{mergeID}$ does not alter the
underlying sets but only aligns their identifiers. Consequently,
$\mathcal{P}_1=\bar{\mathcal{P}}_1$ and $\mathcal{P}_2=\bar{\mathcal{P}}_2$, and
both resulting representations share the identifier vector $\overline {\id}$.

Substituting the representations of $\bar{\mathcal{P}}_1$ and $\bar{\mathcal{P}}_2$
into the definition of the Cartesian product of CPZs yields
\begin{align*}
\mathcal{P}_1 \boxtimes \mathcal{P}_2
=&\bigg \langle \begin{bmatrix} c_{1} \\ c_{2} \end{bmatrix} \begin{bmatrix} G_{1} & \mathbf{0} \\ \mathbf{0} & G_{2} \end{bmatrix}, \begin{bmatrix} E_{1} & \mathbf{0} \\ \mathbf{0} & E_{2} \end{bmatrix},\\& \begin{bmatrix} A_{1} & \mathbf{0} \\ \mathbf{0} & A_{2} \end{bmatrix}, \begin{bmatrix} b_{1} \\ b_{2} \end{bmatrix}, \begin{bmatrix} \CE_{1} & \mathbf{0} \\ \mathbf{0} & \CE_{2} \end{bmatrix}  \bigg \rangle_\text{CPZ},\\=&\bigg\langle
\begin{bmatrix} c_{1} \\ c_{2} \end{bmatrix},
\begin{bmatrix} G_{1} & \mathbf{0} \\ \mathbf{0} & G_{2} \end{bmatrix},
\begin{bmatrix} \bar{E}_{1} & \mathbf{0} \\ \mathbf{0} & \bar{E}_{2} \end{bmatrix}, \\&
\begin{bmatrix} A_{1} & \mathbf{0} \\ \mathbf{0} & A_{2} \end{bmatrix},
\begin{bmatrix} b_{1} \\ b_{2} \end{bmatrix},
\begin{bmatrix} \bar{\CE}_{1} & \mathbf{0} \\ \mathbf{0} & \bar{\CE}_{2} \end{bmatrix},
\overline {\id}
\bigg\rangle_{\text{CPZ}}.
\end{align*}
which concludes the proof.
\end{proof}

We now derive an exact representation for the Cartesian product of two CPZs. This operation constitutes a fundamental building
block for subsequent set-based computations involving polynomial
dynamics.

\begin{proposition}(Exact Multiplication between CPZs) \label{prop:multi1}  
Given $\mathcal{P}_1 = \langle c_{1}, G_{1}, E_{1}, A_{1},b_{1}, \CE_{1} \rangle_{CPZ} \linebreak[3] \subset \Rn$ and $\mathcal{P}_2 = \langle c_{2}, G_{2}, E_{2}, A_{2}, b_{2}, \CE_{2} \rangle_{CPZ} \subset \R^n$, the following identity holds
\begin{multline}
\mathcal{P}_1 \odot \mathcal{P}_2 =  \bigzono{c_{1}\odot c_{2},\begin{bmatrix} G_{1}\odot c_{2} & c_{1} \odot G_{2}& G_{f} \end{bmatrix} ,\\ \hat E, 
    \hat A, \hat {B} , \hat R ,\hat  \id}_\text{CPZ},  \label{eq:cpzcpz}
\end{multline}
where $\mathcal{P}_1\otimes \mathcal{P}_2 \subset \R^{n_x}$ and
\begin{align*}
    \hat A &= \begin{bmatrix} A_{1} & \mathbf{0} \\
     \mathbf{0} & A_{2}
    \end{bmatrix},
    \hat {B} =\begin{bmatrix} b_{1} \\  b_{2}\end{bmatrix} \\
    \hat E &= \bigg[ \overline{E}_{1}, \overline{E}_{2}, 
    \Big[\overline{E}_{1}^{(\cdot,1)} + \overline{E}_{2}^{(\cdot,1)}\Big],
    \dots,
    \Big[\overline{E}_{1}^{(\cdot,1)} + \overline{E}_{2}^{(\cdot,h_{2})}\Big], \\
    & \qquad \dots, \Big[\overline{E}_{1}^{(\cdot,h_{1})} + \overline{E}_{2}^{(\cdot,1)}\Big],
    \dots, \Big[\overline{E}_{1}^{(\cdot,h_{1})}+\overline{E}_{2}^{(\cdot,h_{2})}\Big] \bigg]  
    \\
    \hat {R} &= \begin{bmatrix} \overline{R}_{1} & \overline{R}_{2} \end{bmatrix} \\
    G_{f} &= \begin{bmatrix} g_{f}^{(1)} & \dots & g_{f}^{(h_{1}h_{2})} \end{bmatrix}
\end{align*}
with $\overline{E}_{1}\in\R^{a \times h_{1}}$, $\overline{E}_{2}\in\R^{a \times h_{2}}$, $\overline{R}_{1}\in\R^{a \times q_{1}}$, $\overline{R}_{2}\in\R^{a \times q_{2}}$ and 
$\hat {\id}\in\N^{1 \times a }$ are induced by the 
$\operatorname{mergeID}(\mathcal{P}_1,\mathcal{P}_2)$ operation
as defined in Proposition~\ref{prop:mergeID}.
And
\begin{displaymath}
    g_{f}^{(k)} =  G^{(\cdot,i)}_{1} \odot G_{2}^{(\cdot,j)}, \quad k=h_{2} (i-1) + j,
\end{displaymath}
for $i=1,\dots,h_{1}$,\;$j=1,\dots,h_{2}$.
\end{proposition}
\begin{proof}
Let $\widehat{\mathcal{P}}$ denote the right-hand side of \eqref{eq:cpzcpz}. We prove
$\mathcal{P}_1 \otimes \mathcal{P}_2 \subseteq \widehat{\mathcal{P}}$ and $\widehat{P} \subseteq \mathcal{P}_1 \otimes \mathcal{P}_2$.
Take arbitrary $p_1 \in \mathcal{P}_1$ and $p_2 \in \mathcal{P}_2$. By applying
$\operatorname{mergeID}(\mathcal{P}_1,\mathcal{P}_2)$, $p_1$ and $p_2$ admit
representations with a common identifier vector $\hat \id$, i.e.,
there exists $\hat{\alpha}([1:a])$ such that
\begin{align}
p_1 &= c_{\mathcal{P} _1} + \sum_{i=1}^{h_{1}}
\Big( \prod_{k=1}^{a} \hat{\alpha}_{(k)}^{\overline{E}_{1}(k,i)} \Big)\, G_{1}^{(\cdot,i)}, \label{eq:p1-merged1}\\
p_2 &= c_{P_2} + \sum_{i=1}^{h_{P_2}}
\Big( \prod_{k=1}^{a} \hat{\alpha}_{(k)}^{\overline{E}_{2}(k,i)} \Big)\, G_{2}^{(\cdot,i)}, \label{eq:p2-merged1}
\end{align}
where $a = |\mathcal{H}|+h_{1}$, $\mathcal{H} = \left\{ i~ |~ \id_{{\mathcal{P}_1}(i)} \not\in \id_{2} \right\}$, and the constraints are satisfied simultaneously. From \eqref{mergetwoexact}, we have
    \begin{multline}
    \sum_{i=1}^{q_{1}} \bigg( \prod_{k=1}^{h_{1}}
(\alpha_{1}^{(k)})^{R_{1}^{(k,i)}} \bigg)A_{1}^{{(\cdot,i)}} = \\
    \sum_{i=1}^{q_{1}} \bigg( \prod_{k=1}^{a}
    \hat{\alpha}_{(k)}^{\overline{R}_{1}^{(k,i)}} \bigg)A_{1}^{{(\cdot,i)}}=b_{1}.
    \label{constraint21}
    \end{multline}
    and
    \begin{multline}
    \sum_{i=1}^{q_{2}} \bigg( \prod_{k=1}^{h_{2}}
(\alpha_{2}^{(k)})^{R_{2}^{(k,i)}} \bigg)A_{2}^{{(\cdot,i)}} = \\
    \sum_{i=1}^{q_{2}} \bigg( \prod_{k=1}^{a}
    \hat{\alpha}_{(k)}^{\overline{R}_{2}^{(k,i)}} \bigg)A_{2}^{{(\cdot,i)}}=b_{2}.
    \label{constraint22}
    \end{multline}
Given two sets $S_1, S_2 \subset \mathbb{R}^n$, define the operation as follows:
Element-wise product:
\[
S_1 \odot S_2 := \{\, s_1 \odot s_2 \mid s_1 \in S_1,\ s_2 \in S_2 \,\},
\]
Then, expanding the Hadamard product yields
   \begin{multline}
    p_1 \odot p_2 = c_{1} \odot c_{2} +  \sum_{i=1}^{h_{1}}\bigg( \prod_{k=1}^{a}
    \hat{\alpha}_{(k)}^{\overline{E}_{1}^{(k,i)}} \bigg)  G_{1}^{(\cdot,i)} \odot c_{2} \\
    + c_{1} \odot \sum_{i=1}^{h_{2}} \bigg( \prod_{k=1}^{a}
    \hat{\alpha}_{(k)}^{\overline{E}_{2}^{(k,i)}} \bigg)  G_{2}^{(\cdot,i)} \\
    + \sum_{i=1}^{h_{1}} \sum_{j=1}^{h_{2}} \bigg( \prod_{k=1}^{a}
    \hat{\alpha}_{(k)}^{\overline{E}_{1}^{(k,i)}} \bigg) \bigg( \prod_{k=1}^{a}
    \hat{\alpha}_{(k)}^{\overline{E}_{2}^{(k,j)}} \bigg)  G_{1}^{(\cdot,i)} \odot  G_{2}^{(\cdot,j)} . \label{eq:P_details1}
    \end{multline}
    For the second and third terms on the right-hand side of \eqref{eq:P_details1}, we have two sets of factors, each containing $h_{1}$ and $h_{2}$ elements. Consequently, the first $h_{1} + h_{2}$ entries of $\hat{\alpha}$ and the columns of $\hat E$ are defined accordingly as follows:
    \begin{align}
    \hat{\alpha}_{([1:h_{1} + h_{2}])} &=\bigg[\prod_{k=1}^{a}
    \hat{\alpha}_{(k)}^{\overline{E}_{1}^{(k,1)}}, \dots,\prod_{k=1}^{a}
    \hat{\alpha}_{(k)}^{\overline{E}_{1}^{(k,h_{1})}} \nonumber\\
    & \hspace{1cm} 
    \prod_{k=1}^{a} \hat{\alpha}_{(k)}^{\overline{E}_{2}^{(k,1)}}, \dots, 
    \prod_{k=1}^{a} \hat{\alpha}_{(k)}^{\overline{E}_{2}^{(k,h_{2})}}\bigg] \\ \hat {E}^{([1:h_{1} + h_{2}])} &= \Big[ \overline{E}_{1},\overline{E}_{2}\Big].
\end{align} 
Because of $\operator{mergeID}(\mathcal{P}_1,\mathcal{P}_2)$, the fourth term on the right-hand side of \eqref{eq:P_details1} can be expressed as:
\begin{multline}
\sum_{i=1}^{h_{1}} \sum_{j=1}^{h_{2}} \bigg( \prod_{k=1}^{a}
  \hat{\alpha}_{(k)}^{\overline{E}_{1}^{(k,i)}}  
  \hat{\alpha}_{(k)}^{\overline{E}_{2}^{(k,j)}} \bigg)  G_{1}^{(\cdot,i)} \odot G_{2}^{(\cdot,j)}= \\
  \sum_{i=1}^{h_{1}} \sum_{j=1}^{h_{2}} \bigg( \prod_{k=1}^{a}
  \hat{\alpha}_{(k)}^{\overline{E}_{1}^{(k,i)}+  
  \overline{E}_{2}^{(k,j)}} \bigg)  G_{1}^{(\cdot,i)}  \odot G_{2}^{(\cdot,j)}. \label{1:n+p1}
\end{multline} 
Concatenating the factors in \eqref{1:n+p1}, we have
\begin{multline*}
\hat{\alpha}_{([h_{1}+h_{2}+1:h_{1}+h_{2}+h_{1} h_{2}])} \\ 
=\bigg[\prod_{k=1}^{a}
  \hat{\alpha}_{(k)}^{\overline{E}_{1}^{(k,1)}+  
  \overline{E}_{2}^{(k,1)}}, \dots, \prod_{k=1}^{a}
  \hat{\alpha}_{(k)}^{\overline{E}_{1}^{(k,1)}+  
  \overline{E}_{2}^{(k,h_{2})}}, \dots, \\
  \prod_{k=1}^{a}
  \hat{\alpha}_{(k)}^{\overline{E}_{1}^{(k,h_{1})}+  \overline{E}_{2}^{(k,1)}}, \dots,\prod_{k=1}^{a}
  \hat{\alpha}_{(k)}^{\overline{E}_{1}^{(k,h_{1})}+  \overline{E}_{2}^{(k,h_{2})}}\bigg],
\end{multline*}
 which results in $\hat E^{([h_{1}+h_{2}+1:h_{1}+h_{2}+h_{1} h_{2}])}$ and $G_{f}$ as follows:
 \begin{align*}
    &\hat E^{([h_{1}+h_{2} + 1 : h_{1}+h_{2} + h_{1} h_{2}])}= \bigg[  
    \Big[\overline{E}_{1}^{(\cdot,1)} + \overline{E}_{2}^{(\cdot,1)}\Big],
    \dots, \\
    &
    \Big[\overline{E}_{1}^{(\cdot,1)}  + \overline{E}_{2}^{(\cdot,h_{2})}\Big], 
    \dots, \Big[\overline{E}_{1}^{(\cdot,h_{1})} + \overline{E}_{2}^{(\cdot,1)}\Big],\dots, \\
    &\Big[\overline{E}_{1}^{(\cdot,h_{1})}+\overline{E}_{2}^{(\cdot,h_{2})}\Big] \bigg] \\
    &G_{f} = \bigg[ G^{(\cdot,1)}_{1} G_{2}^{(\cdot,1)},{...}, G^{(\cdot,1)}_{1} G_{2}^{(\cdot,h_{2})},{...},G^{(\cdot, h_{1})}_{\mathcal{Y}} G_{\mathcal{P}}^{(\cdot,h_{2})} \bigg].
    \end{align*}
Thus, $p_1 \odot p_2 \in \hat{\mathcal{P}}$ and therefore $\mathcal{P}_1 \odot \mathcal{P}_2 \subseteq \hat{\mathcal{P}}$. 
    Conversely, let $\hat{p} \in \hat{\mathcal{P}}$, then there exists $ \hat{\alpha}_{([1:{h_{1}+h_{2}+h_{1} h_{2}}])}:$
    \begin{align*}
      \hat p =\nonumber\hat c + \sum_{i=1}^{h_{1}+h_{2}+h_{1} h_{2}}\bigg( \prod_{k=1}^{a}\hat{\alpha}_{(k)}^{\hat {E}^{(k,i)}} \bigg)  G_{f}^{(i)}.
    \end{align*}
    By partitioning
    \begin{multline*}
    \hat{\alpha}_{([1:{h_{1}+h_{2}+h_{1} h_{2}}])}=\bigg[\prod_{k=1}^{a}
    \hat{\alpha}_{(k)}^{\overline{E}_{1}^{(k,[1:h_{1}])}},\prod_{k=1}^{a}
    \hat{\alpha}_{(k)}^{\overline{E}_{2}^{(k,[1:h_{2}])}},\\
    \prod_{k=1}^{a}
    \hat{\alpha}_{(k)}^{\overline{E}_{1}^{(k,1)}+  
    \overline{E}_{2}^{(k,1)}}, \dots,\prod_{k=1}^{a}
    \hat{\alpha}_{(k)}^{\overline{E}_{1}^{(k,1)}+  
    \overline{E}_{2}^{(k,h_{2})}},\dots,\\
    \prod_{k=1}^{a}
    \hat{\alpha}_{(k)}^{\overline{E}_{1}^{(k,h_{1})}+  \overline{E}_{2}^{(k,1)}}, \dots,\prod_{k=1}^{a}
    \hat{\alpha}_{(k)}^{\overline{E}_{1}^{(k,h_{1})}+  \overline{E}_{2}^{(k,h_{2})}}\bigg]  
    \label{eq:beta_c1_3}
    \end{multline*}
    it follows that there exist $p_1\in \mathcal{P}_1$ and $p_2 \in \mathcal{P}_2$ such that $\hat{p} = p_1 \odot p_2$.
    Meanwhile, since $\hat{p}\in\hat{\mathcal{P}}$, it holds that
    \begin{align}
    \sum_{i=1}^{q_{2}+q_{1}} \bigg( \prod_{k=1}^{a}
    \hat{\alpha}_{(k)}^{\hat {R}^{(k,i)}} \bigg)\hat A^{{(\cdot,i)}}=\begin{bmatrix}  b_{1} \\  b_{2}\end{bmatrix}
    \end{align}
    which satisfies the constraints in \eqref{constraint3}. Therefore, $\hat{p} \in \mathcal{P}_1 \odot \mathcal{P}_2$ and thus $ \hat{\mathcal{P}} \subseteq \mathcal{P}_1  \odot  \mathcal{P}_2$.
\end{proof}

\Comp The quantities \(\hat B\), \(\hat A\), \(\hat E\), \(c_{1}\odot c_{2}\), \(G_{1}\odot c_{2}\), \(c_{1}\odot G_{2}\), and \(G_f\) are obtained using standard block concatenation, Hadamard products, and matrix multiplications. In particular, \(\hat B\) has complexity \(\mathcal{O}(m_{1}+m_{2})\), \(\hat A\) has complexity \(\mathcal{O}((m_{1}+m_{2})(q_{1}+q_{2}))\), \(\hat E\) has complexity \(\mathcal{O}((|\mathcal{H}|+h_{1})h_{1}h_{2})\), \(c_{1}\odot c_{2}\) has complexity \(\mathcal{O}(n)\), \(G_{1}\odot c_{2}\) and \(c_{1}\odot G_{2}\) have complexities \(\mathcal{O}(n h_{1})\) and \(\mathcal{O}(n h_{2})\), respectively, and \(G_f\) has complexity \(\mathcal{O}(n h_{1}h_{2})\). Therefore, the exact multiplication in \(\eqref{eq:cpzcpz}\) has overall complexity \(\mathcal{O}\!\bigl((m_{1}+m_{2})(q_{1}+q_{2})+(|\mathcal{H}|+h_{1})h_{1}h_{2}+n(h_{1}+h_{2})+n h_{1}h_{2}\bigr)\). \hfill \(\diamond\)

\begin{proposition}[Subset projection]
\label{prop:coordinate_extraction}
Let $\mathcal{Z} \subset \mathbb{R}^n$ be a (constrained polynomial) zonotope describing a set of vectors
\[
x = \begin{bmatrix} x_1 & x_2 & \dots & x_n \end{bmatrix}^T.
\]
Let $\mathcal{I} = \{ i_1, i_2, \dots, i_k \} \subseteq \{1,\dots,n\}$ be an arbitrary index set with $k \leq n$.
Define the linear map
\[
E_{\mathcal{I}} : \mathbb{R}^n \rightarrow \mathbb{R}^k,
\qquad
E_{\mathcal I} x = \big[\, x_{i_1},\, x_{i_2},\,\ldots,\,x_{i_k}\,\big]^\top .
\]
Equivalently, $E_{\mathcal{I}}$ can be represented by a matrix in $\mathbb{R}^{k\times n}$ obtained by selecting the rows of the identity matrix $I_n$ indexed by $\mathcal{I}$.
Then, the image of $\mathcal{Z}$ under $E_{\mathcal{I}}$ is given by
\[
\mathcal{Z}_{\mathcal{I}} = E_{\mathcal{I}}(\mathcal{Z})
= \left\{ E_{\mathcal{I}}x \in \mathbb{R}^k \;\middle|\; x \in \mathcal{Z} \right\},
\]
that is, $\mathcal{Z}_{\mathcal{I}}$ consists exactly of all subvectors of $x$ formed by the components indexed by $\mathcal{I}$.
Moreover, if $\mathcal{Z}$ is a constrained polynomial zonotope, then $\mathcal{Z}_{\mathcal{I}}$ is also a constrained polynomial zonotope of dimension $k$, and no approximation is introduced by this operation.
\end{proposition}

Based on the results established above, Algorithm~\ref{alg:PolyReachability}
summarizes the proposed model-based reachable set computation method
for polynomial systems.
The algorithm implements the exact set propagation of the polynomial dynamics
by evaluating the monomial vector $h(\cdot)$ at the set level and applying the
corresponding linear mapping.
The reachable sets computed by
Algorithm~\ref{alg:PolyReachability} are exact and do not rely on outer approximations.

\subsection{Data-driven Reachability Analysis for Polynomial System}
Similarly to the definition of $\mathcal{N}_\Sigma$ in \eqref{eq:Nsig}, we denote the set of unknown coefficients consistent with the data, including the true coefficients $\Theta_{\text{tr}}$, by $\mathcal{N}^{\Sigma_p}$. 
From \eqref{eq:pcg}, let
\begin{align*}
\Omega=\begin{bmatrix} h(x(0),u(0)) \, \dots \, h(x(T-1),u(T-1))\end{bmatrix}
\end{align*}
Then, the following result computes a set of coefficients that is consistent with the data and includes the true coefficients~$\Theta_{\text{tr}}$.
\begin{algorithm}[t!]
  \caption{Model-based Polynomial Reachable Set Computation}
  \label{alg:PolyReachability}
  \textbf{Input}: coefficient matrix $\Theta_{\text{tr}}$ of the polynomial system in \eqref{eq:sysnonlin_poly}, initial set $\mathcal{X}_{0}$, process noise zonotope $\mathcal{Z}_w$, and input zonotope $\mathcal{U}_k$, $\forall  k \in \{0,1,\dots,N-1\}$. \\
  \textbf{Output}: reachable sets $\hat{\mathcal{R}}_{k}^p$, $\forall  k \in \{0,1,\dots,N\}$
 \begin{algorithmic}[1]
  \State $\hat{\mathcal{R}}_{0}^p = \mathcal{X}_{0}$
  \For{$k = 0:N-1$}
    \State $\mathcal{Z}_k = \hat{\mathcal{R}}_k^p \boxtimes \mathcal{U}_k$
    \For{$\ell = 1:n_z$}
        \State $\mathcal{Z}_{k,\ell} = E_\ell(\mathcal{Z}_k)$
    \EndFor
    \For{$j = 1:m_a$}
        \State $\mathcal{M}_{k,j} = \{1\}$
        \For{$\ell = 1:n_z$}
            \State $\mathcal{M}_{k,j} = \mathcal{M}_{k,j} \odot (\mathcal{Z}_{k,\ell})^{\alpha^{(j)}_\ell}$
        \EndFor
    \EndFor
    \State $h(\mathcal{Z}_k) = \mathcal{M}_{k,1} \boxtimes \mathcal{M}_{k,2} \boxtimes \dots \boxtimes \mathcal{M}_{k,m_a}$
    \State $\hat{\mathcal{R}}_{k+1}^p = \Theta_{\text{tr}} \, h(\mathcal{Z}_k) + \mathcal{Z}_w$
  \EndFor
 \end{algorithmic}
\end{algorithm}

\begin{algorithm}[t!]
\caption{Data-driven Polynomial Reachable Set Computation with Online Update and Intersection}
\label{alg:DDPolyReachability_update_intersect}
\textbf{Input:}
Initial input--state data $D_0=(X_0^+,X_0^-,U_0^-)$,
initial set $\mathcal{X}_0$,
process noise zonotope $\mathcal{Z}_w$,
matrix zonotope $\mathcal{M}_w$, batch size $L \ge m_a$, input zonotopes $\mathcal{U}_k$, $\forall k \in \{0,1,\dots,N-1\}$. \\
\textbf{Output:}
Reachable sets $\hat{\mathcal{R}}_k^p$, $\forall  k \in \{0,1,\dots,N\}$
\begin{algorithmic}[1]
\State $\hat{\mathcal{R}}_0^p = \mathcal{X}_0$
\State $\Omega_0 = \begin{bmatrix} h(x(0),u(0)) & \dots & h(x(T_0-1),u(T_0-1)) \end{bmatrix}$
\State $\mathcal{M}_0^{\Sigma_p} = (X_0^+ - \mathcal{M}_w)\,\Omega_0^\dagger$
\State $\hat{\mathcal{M}}^{\Sigma_p} = \mathcal{M}_0^{\Sigma_p}$
\State $X^+ = [~],\;\Omega = [~]$
\For{$k = 0:N-1$}
    \State $X^+ = [X^+~ x_{(k)}]$
    \State $\Omega = [\Omega~ h(x_{(k-1)},u_{(k-1)})]$
    \If{$\mathrm{cols}(\Omega)\ge L$ \textbf{and} $\mathrm{rank}(\Omega)=m_a$}
        \State $\mathcal{M}^{\Sigma_p} = (X^+ - \mathcal{M}_w)\,\Omega^\dagger$
        \State $\hat{\mathcal{M}}^{\Sigma_p} = \hat{\mathcal{M}}^{\Sigma_p} \cap \mathcal{M}^{\Sigma_p}$
        \State $X^+ = [~],\;\Omega = [~]$
    \EndIf
    \State $\mathcal{Z}_k = \hat{\mathcal{R}}_k^p \boxtimes \mathcal{U}_k$
    \For{$\ell = 1:n_z$}
        \State $\mathcal{Z}_{k,\ell} = E_\ell(\mathcal{Z}_k)$
    \EndFor
    \For{$j = 1:m_a$}
        \State $\mathcal{M}_{k,j}^{\Sigma_p} = \{1\}$
        \For{$\ell = 1:n_z$}
            \State $\mathcal{M}_{k,j}^{\Sigma_p} = \mathcal{M}_{k,j}^{\Sigma_p} \odot (\mathcal{Z}_{k,\ell})^{\alpha^{(j)}_\ell}$
        \EndFor
    \EndFor
    \State $h(\mathcal{Z}_k) = \mathcal{M}_{k,1}^{\Sigma_p} \boxtimes \mathcal{M}_{k,2}^{\Sigma_p} \boxtimes \dots \boxtimes \mathcal{M}_{k,m_a}^{\Sigma_p}$
    \State $\hat{\mathcal{R}}_{k+1}^p = \hat{\mathcal{M}}^{{\Sigma_p}}\, h(\mathcal{Z}_k) + \mathcal{Z}_w$
\EndFor
\end{algorithmic}
\end{algorithm}



\begin{lemma}[\cite{Alanwar2023Datadriven}]
\label{lm:sigmaM_p}
Given a matrix $\Omega$ of the polynomial system in \eqref{eq:sysnonlin_poly} with a full row rank, then the matrix zonotope 
\begin{align}
    \mathcal{M}^{\Sigma_p} = (X_{+} - \mathcal{M}_w) \Omega^\dagger \label{eq:Msigma}
\end{align} 
 contains all matrices $\Theta$ that are consistent with the data and the noise bound, i.e., $\mathcal{M}^{\Sigma_p} \supseteq \mathcal{N}^{\Sigma_p}$. 
\end{lemma}

\begin{remark}
The condition in Lemma~\ref{lm:sigmaM_p} of requiring $\Omega$ with a full row rank implies that there exists a right-inverse of the matrix $\Omega$. This condition can be easily checked given the data. 
\end{remark}


\begin{proposition}
\label{prop:reach_poly}
Consider a sequence of input--state data sets
$D_i = (X_i^+, \Omega_i)$ for $i \geq 0$, where $D_0$ denotes an initial data set
and $D_1, D_2, \dots$ are subsequently collected data generated by the
polynomial system \eqref{eq:sysnonlin_poly}.
Suppose that each data set $D_i$ satisfies Assumption~\ref{ass:rank_D},
and induces a corresponding set of polynomial models
$\mathcal{M}_i^{p}$ according to Lemma~\ref{lm:sigmaM_p}.
Define the recursively refined model set as
\begin{align}
\hat{\mathcal{M}}_i^{\Sigma_p}
=
\mathcal{M}_i^{\Sigma_p} \cap \hat{\mathcal{M}}_{i-1}^{\Sigma_p},
\quad i \geq 1,
\qquad
\hat{\mathcal{M}}_0^{\Sigma_p}
=
\mathcal{M}_0^{\Sigma_p}.
\label{eq:poly_model_intersection}
\end{align}
Then, the set $\hat{\mathcal{M}}_i^{\Sigma_p}$ is a constrained polynomial zonotope
and contains the true polynomial model matrix $\Theta_{\mathrm{tr}}$.
\end{proposition}

After identifying a set of polynomial coefficients that is consistent with the
available data, the remaining challenge is the forward propagation of the
reachable set.
In the linear case, this task can
be addressed using linear mappings and Minkowski sum operations, which are
naturally supported by the zonotope representation.
For polynomial systems, however, forward propagation requires the evaluation of
monomial terms of the reachable sets, as shown in \eqref{eq:pcg}, which cannot be
carried out exactly within the class of standard zonotopes.

Existing data-driven approaches for polynomial systems, such as
\cite{Alanwar2023Datadriven}, address this difficulty by over-approximating the
reachable and input sets with interval representations, enabling monomial
evaluation via interval arithmetic.
While this strategy ensures computational tractability, it introduces additional
conservatism and prevents exact set propagation.

In contrast, the approach proposed in this work exploits the algebraic structure
of CPMZs and CPZs.
By leveraging exact multiplication operations between CPMZs and CPZs, as well as
exact multiplication between CPZs, monomials can be evaluated directly
at the set level without resorting to interval over-approximations.
As a consequence, both model-based and data-driven polynomial systems admit an
exact reachable set propagation within the zonotopic framework.

The algorithm is summarized in Algorithm~\ref{alg:DDPolyReachability_update_intersect}.
We first initialize the reachable set $\hat{\mathcal{R}}_0^p$ and compute an initial
set of models $\hat{\mathcal{M}}^{\Sigma_p}$ that is consistent with the
available input--state data.
Subsequently, as new data become available, the model set is refined online through
recursive set intersection.
At each time step $k=0,\dots,N-1$, the reachable set is propagated forward by
evaluating the polynomial dynamics directly at the set level.
To this end, the reachable set and the input set are first combined into the joint
set $\mathcal{Z}_k = \hat{\mathcal{R}}_k^p \boxtimes \mathcal{U}_k$.
The monomial vector $h(\mathcal{Z}_k)$ is then computed exactly by coordinate
extraction and exact set multiplication within the class of CPZs.
Finally, in the last step of the iteration, the reachable set is updated by applying
the refined model set $\hat{\mathcal{M}}^{\Sigma_p}$ and incorporating
the process noise zonotope $\mathcal{Z}_w$, yielding the next reachable set
$\hat{\mathcal{R}}_{k+1}^p$.

\begin{remark}
In Algorithm~\ref{alg:DDPolyReachability_update_intersect}, the model set construction in Line~3 and the subsequent refinement via intersection in Line~10 can be enhanced in the same manner as in the LTI case; see Remark~\ref{rem:cmz_sideinfo_lti}.
\end{remark}

\begin{figure*}[t]
    \centering
    \begin{subfigure}[h]{0.32\textwidth}
        \includegraphics[width=\linewidth, trim=45 45 45 45, clip]{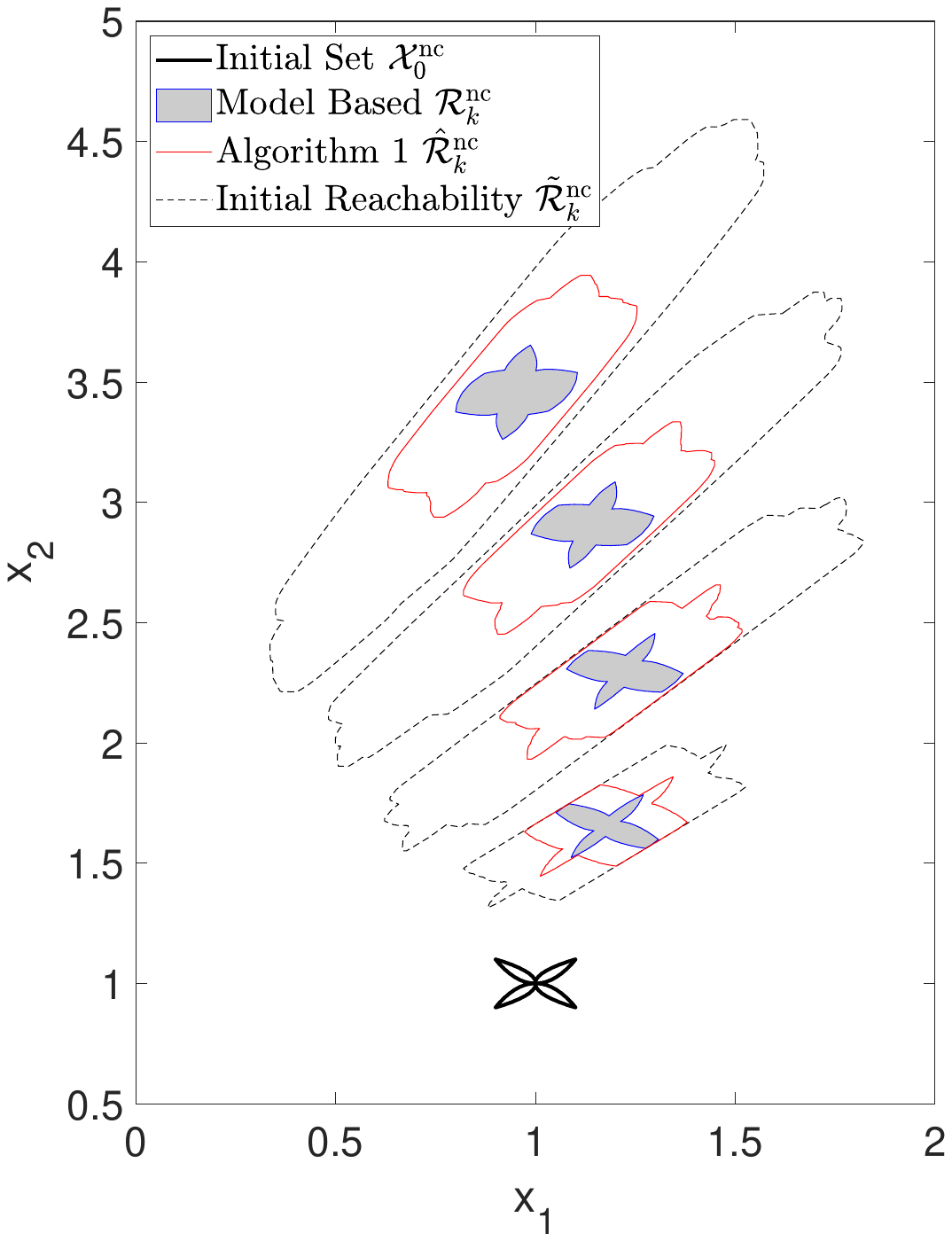}
        \caption{}
        \label{fig:x1x2_a}
    \end{subfigure}
    \begin{subfigure}[h]{0.32\textwidth}
        \includegraphics[width=\linewidth, trim=45 45 45 45, clip]{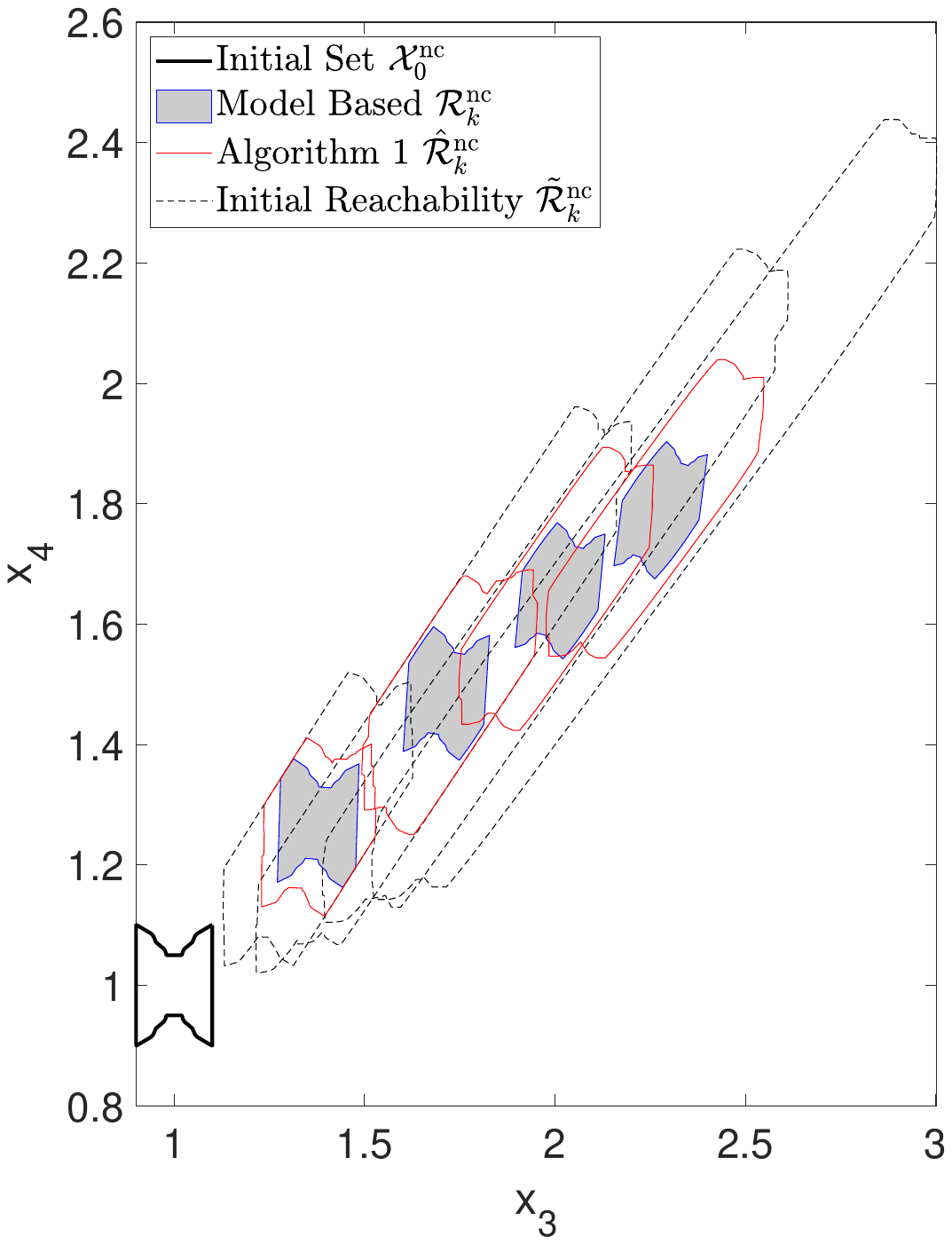}
        \caption{}
        \label{fig:x3x4_a}
    \end{subfigure}
    \begin{subfigure}[h]{0.32\textwidth}
        \includegraphics[width=\linewidth, trim=45 45 45 45, clip]{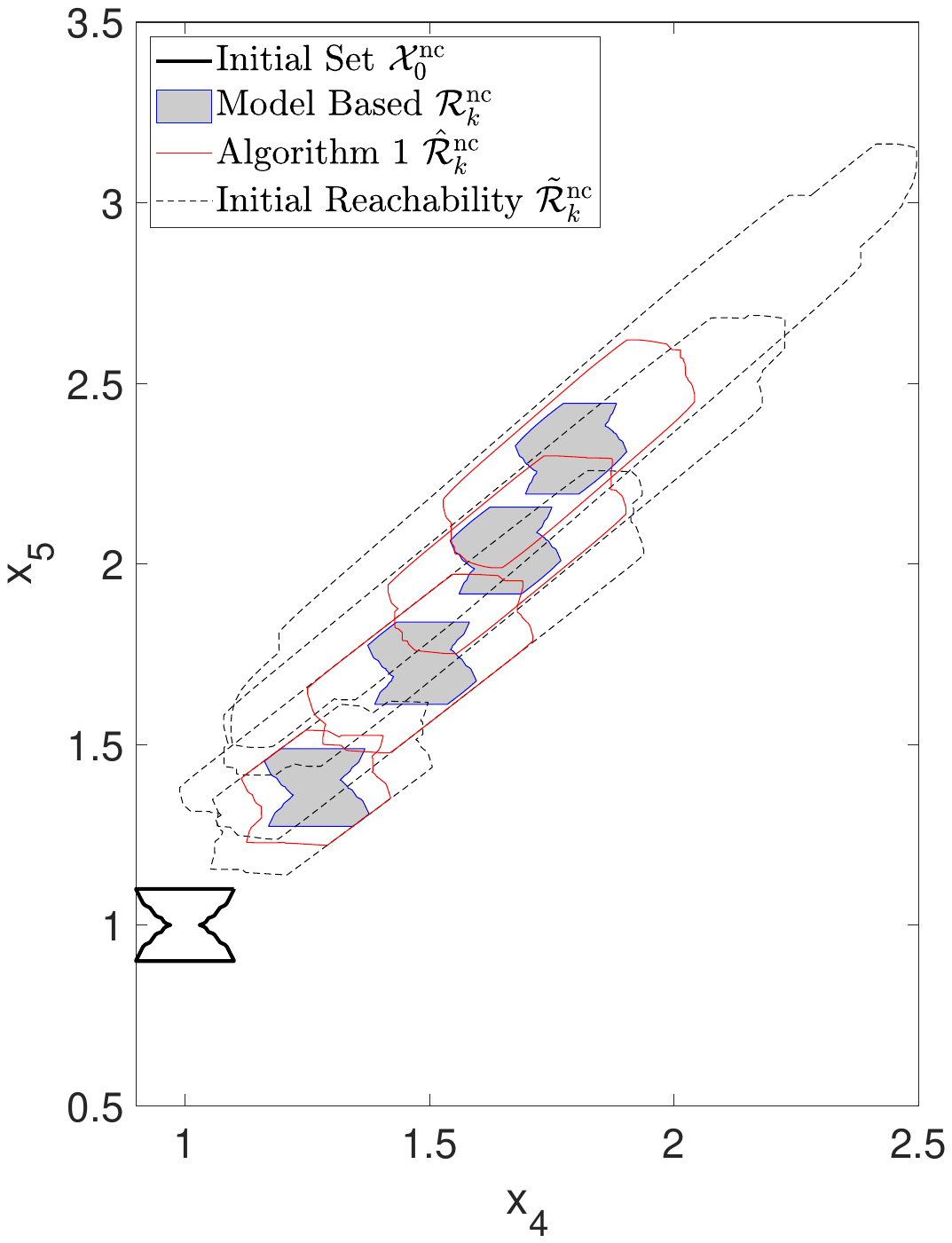}
        \caption{}
        \label{fig:x4x5_a}
    \end{subfigure}
    \caption{Projection of Reachable Sets Computed from Input--State Data Using Algorithm \ref{alg:offline-online} with a Nonconvex Initial Set}
    \label{fig:projSetA}
\end{figure*}
\begin{figure*}[!htbp]
    \centering
    \begin{subfigure}[h]{0.32\textwidth}
        \includegraphics[width=\linewidth, trim=45 45 45 45, clip]{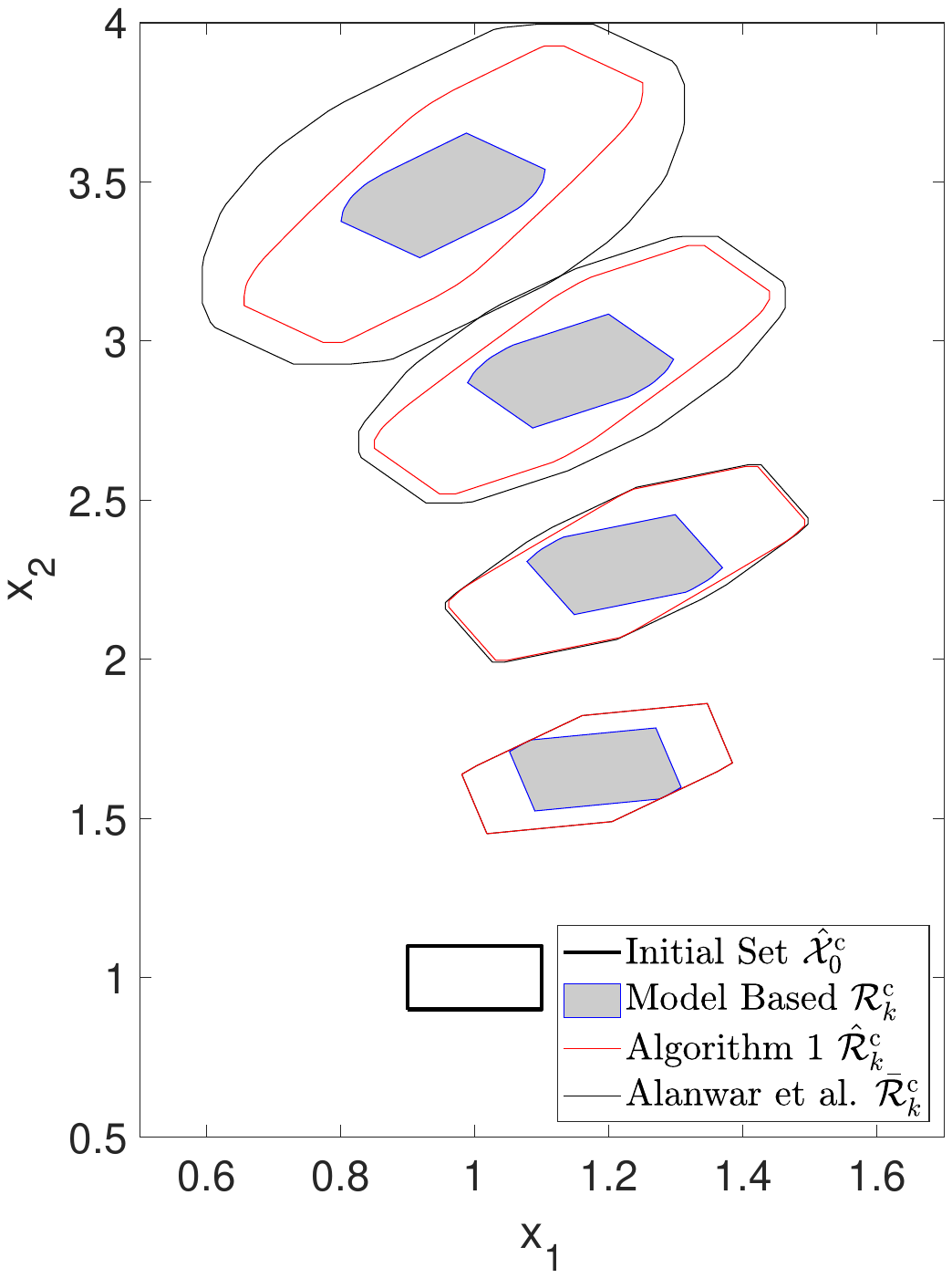}
        \caption{}
        \label{fig:x1x2_b}
    \end{subfigure}
    \begin{subfigure}[h]{0.32\textwidth}
        \includegraphics[width=\linewidth, trim=45 45 45 45, clip]{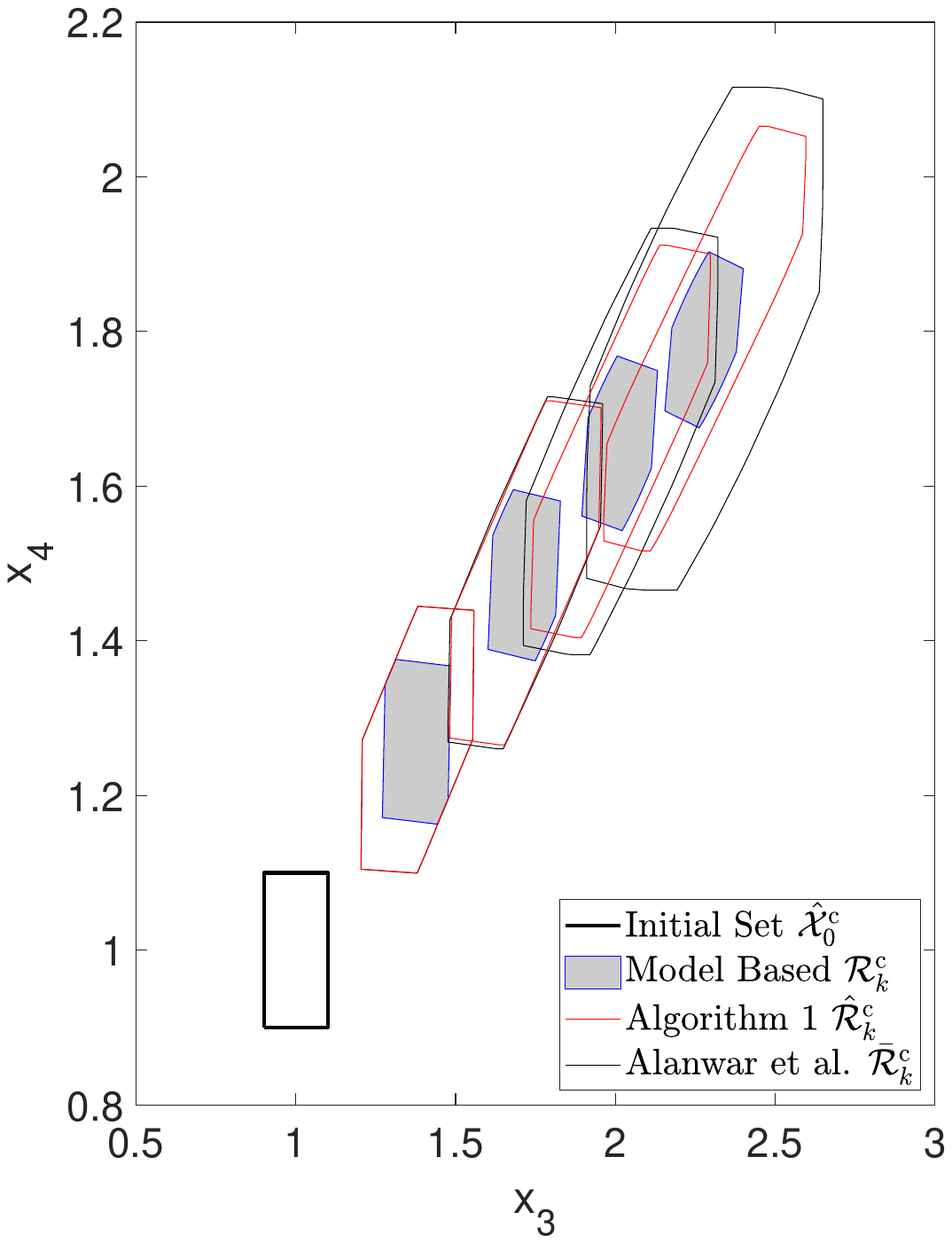}
        \caption{}
        \label{fig:x3x4_b}
    \end{subfigure}
    \begin{subfigure}[h]{0.32\textwidth}
        \includegraphics[width=\linewidth, trim=45 45 45 45, clip]{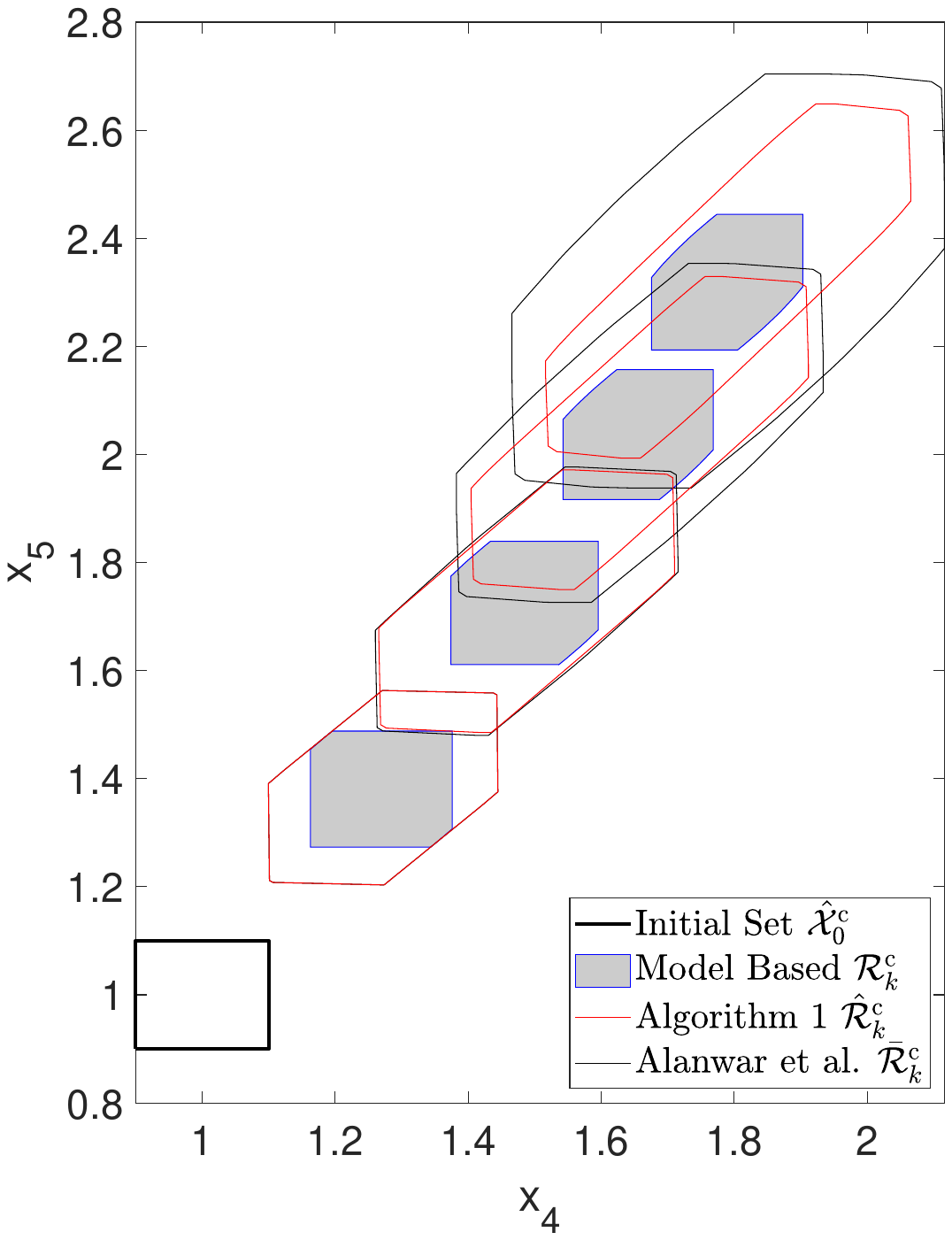}
        \caption{}
        \label{fig:x4x5_b}
    \end{subfigure}
    \caption{Comparative Projection of Reachable Sets from Input--State Data Using Algorithm \ref{alg:offline-online} and Alanwar et al.~\cite{Alanwar2023Datadriven}}
    \label{fig:projSetB}
\end{figure*}
\begin{figure*}[!htbp]
    \centering
    \begin{subfigure}[h]{0.32\textwidth}
        \includegraphics[width=\linewidth, trim=45 45 45 45, clip]{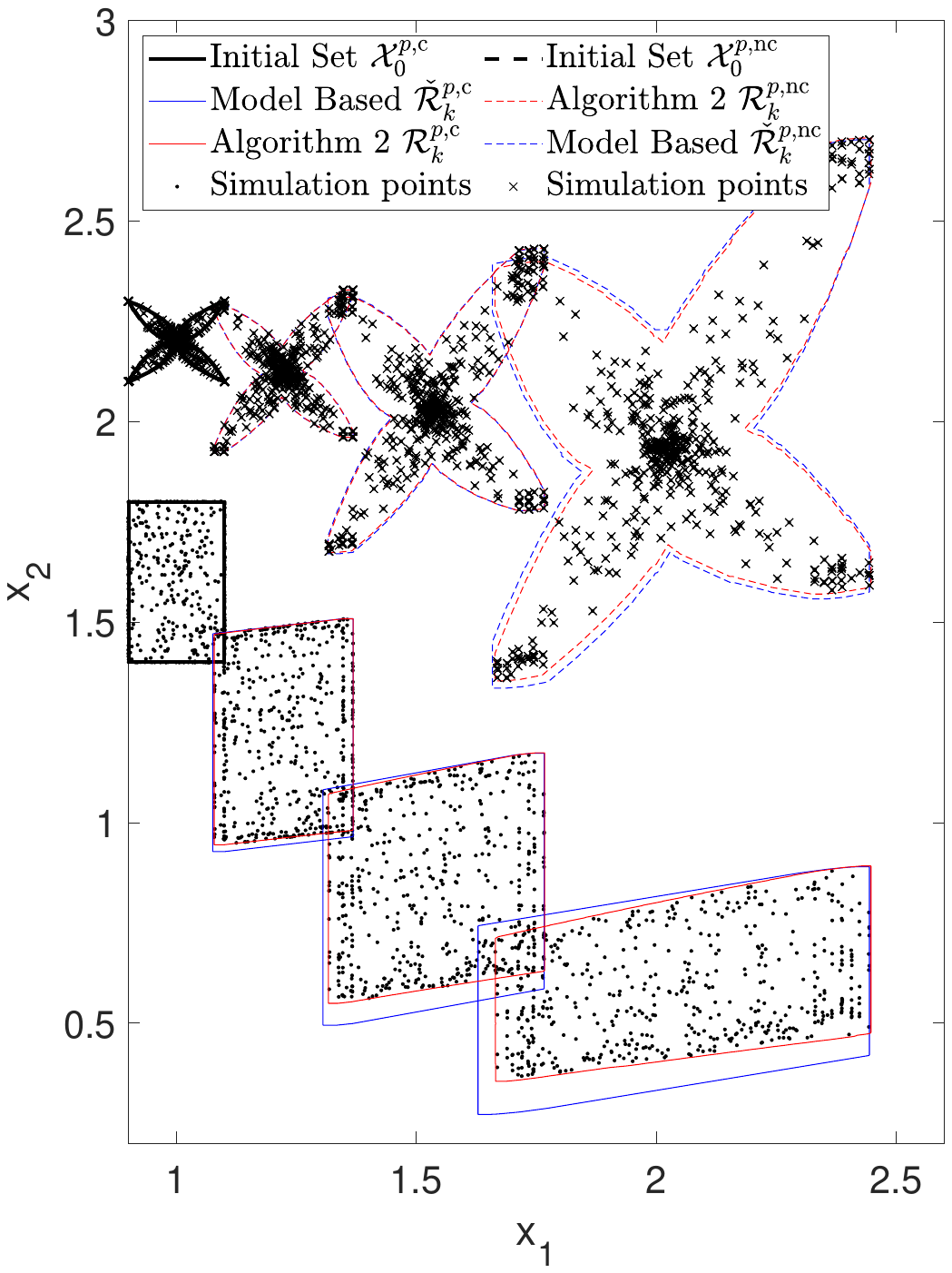}
        \caption{}
        \label{fig:x1x2_c}
    \end{subfigure}
    \begin{subfigure}[h]{0.32\textwidth}
        \includegraphics[width=\linewidth, trim=45 45 45 45, clip]{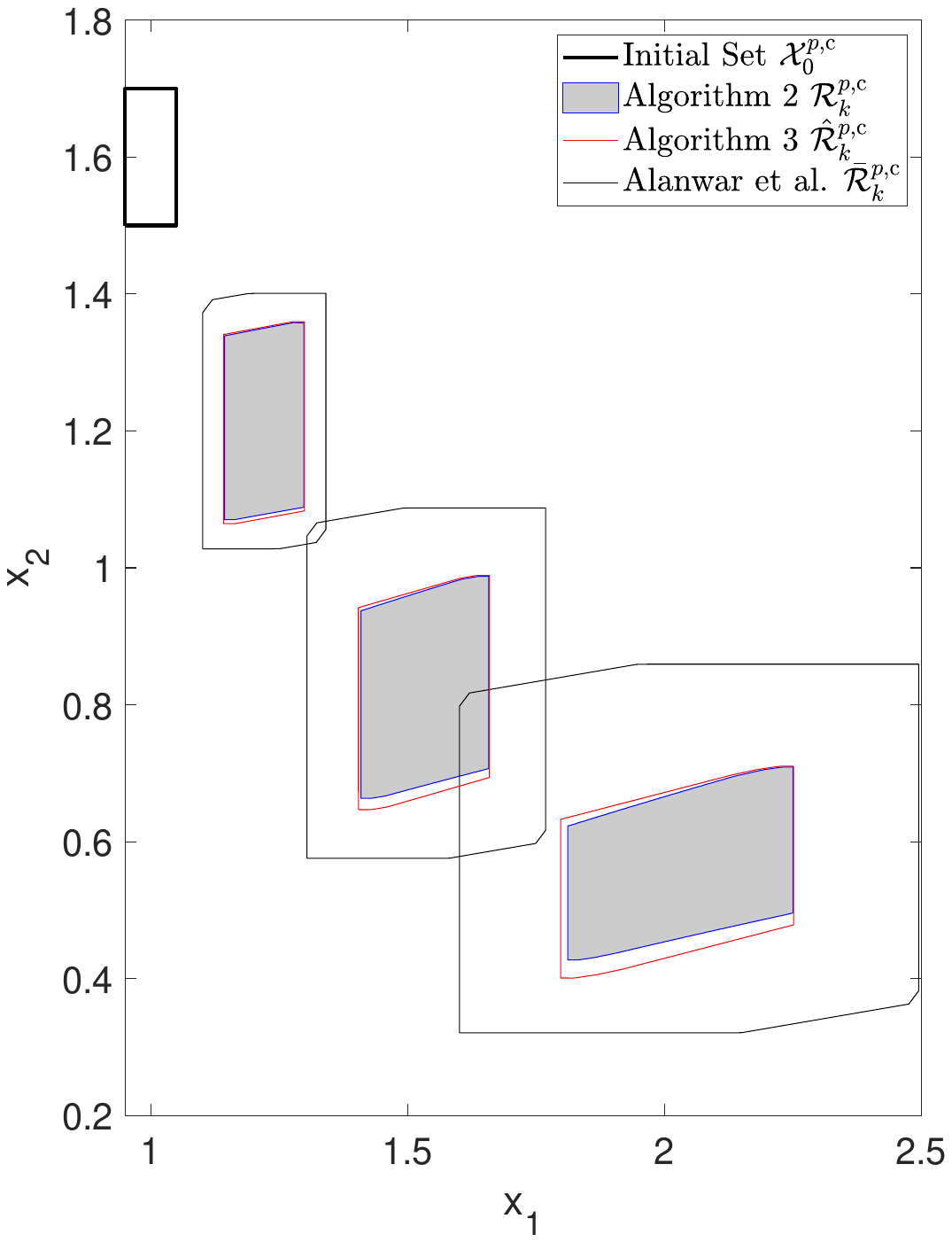}
        \caption{}
        \label{fig:x3x4_c}
    \end{subfigure}
    \begin{subfigure}[h]{0.32\textwidth}
        \includegraphics[width=\linewidth, trim=45 45 45 45, clip]{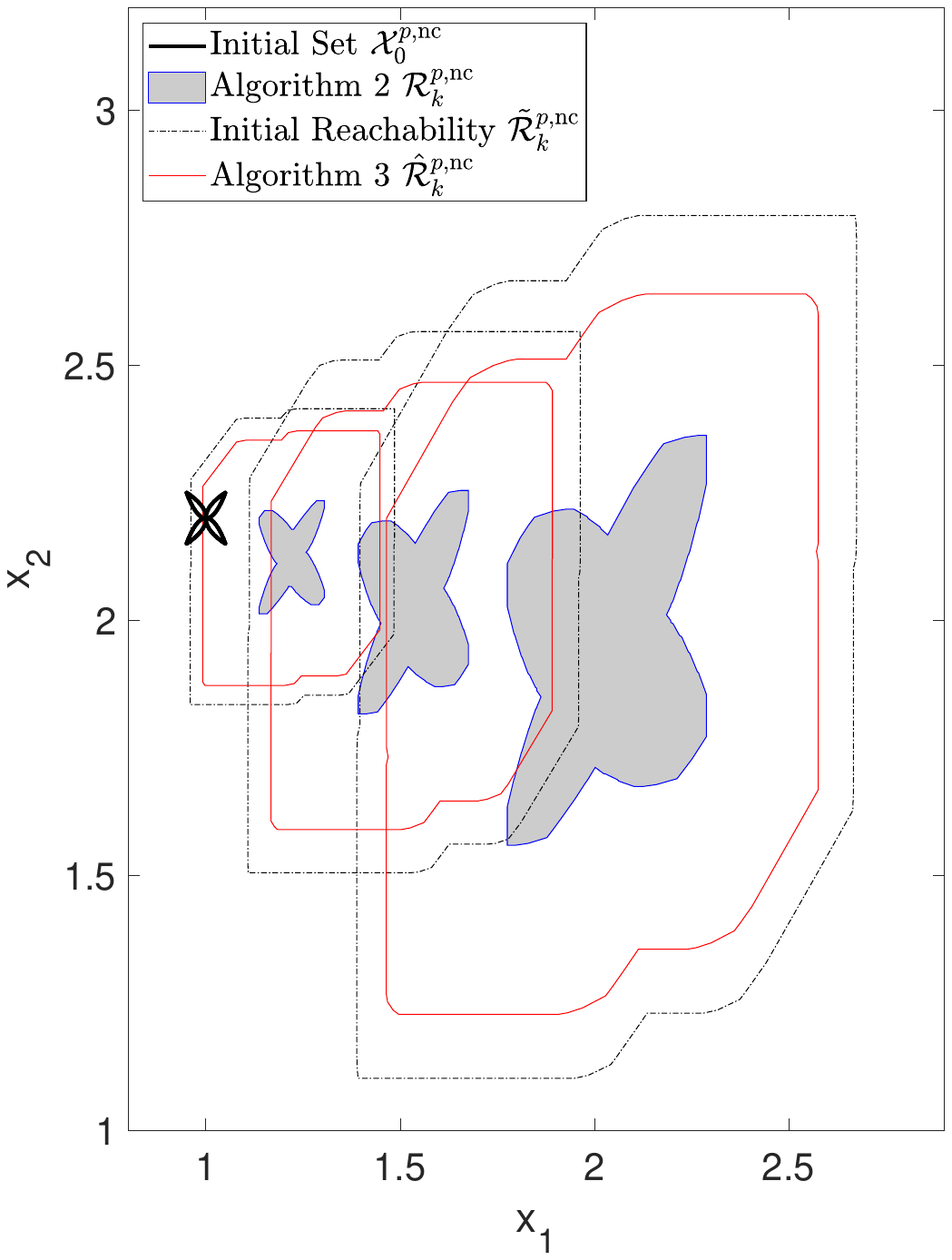}
        \caption{}
        \label{fig:x4x5_c}
    \end{subfigure}
   \caption{Projections of reachable sets for a polynomial system.
(a) Model-based reachable sets computed using
Algorithm~\ref{alg:PolyReachability}, compared with the model-based reachability
analysis results obtained using CORA~\cite{althoff2015introduction}, for both convex and nonconvex initial
sets, together with Monte Carlo simulation trajectories.
(b) Data-driven reachable sets computed using
Algorithm~\ref{alg:DDPolyReachability_update_intersect} for a convex initial set
under a small disturbance magnitude, compared with the interval-based
data-driven method of Alanwar et al.~\cite{Alanwar2023Datadriven}.
(c) Data-driven reachable sets for a nonconvex initial set under a larger
disturbance magnitude, illustrating the effect of online set refinement.}

    \label{fig:projSetB1}
\end{figure*}



\section{Numerical simulations}\label{sec:numerical-simulations}


\subsection{LTI Systems}
Consider a five-dimensional system which is a discretization of the system used in \cite{Alanwar2023Datadriven} with sampling time $0.05$ sec.  All used code to reproduce our results are publicly available\footnotemark. 
The system has the following matrices.
\footnotetext{\href{https://github.com/TUM-CPS-HN/Data-Driven-Nonconvex-Reachability-Analysis}{https://github.com/TUM-CPS-HN/Data-Driven-Nonconvex-Reachability-Analysis}}
\begin{align*}
\Phi_\tr&=\begin{bmatrix}
    0.9323 &  -0.1890   &      0   &      0   &      0 \\
    0.1890 &   0.9323  &       0  &       0   &      0 \\
         0 &        0  &  0.8596  &   0.0430  &        0 \\
         0 &         0   & -0.0430    & 0.8596      &    0 \\
         0 &         0  &        0    &      0   &  0.9048
\end{bmatrix},\\
\Gamma_\tr&=\begin{bmatrix} 
    0.0436&
    0.0533&
    0.0475&
    0.0453&
    0.0476
    \end{bmatrix}^\top.
\end{align*}

In the first experiment, the initial nonconvex set is chosen to be $\mathcal{X}_0^{\mathrm{nc}} = \zono{c_0,G_0,E_0,[~],[~],[~],\id_0}_{\text{CPZ}}$ , where
\begin{align*}
\begin{split}
c_0=1_{5 \times1},~ G_0=0.1I_5,~ E_0&=\begin{bmatrix}
    2 &  1   &      0   &      0   &      0 \\
    1 &   2  &       0  &       0   &      0 \\
         0 &        0  &  2  &   1  &        0 \\
         0 &         0   & 1    & 2      &    1 \\
         0 &         0  &        0    &      1   &  2
\end{bmatrix}.
\end{split}
\end{align*}
For the second experiment, a convex initial set $\hat{\mathcal{X}}_0^{\mathrm{c}}=\zono{1_{5 \times1},0.1 I_5}$ is used to compare reachable sets computed via Algorithm~\ref{alg:offline-online} and the method of Alanwar et al.~\cite{Alanwar2023Datadriven}. In both experiments, the input set is defined as $\mathcal{U}_k = \zono{10, 0.25}$, and random noise is sampled from the zonotope $\mathcal{Z}_w = \zono{0, [0.005 , \dots , 0.005]^\top}$. 


Firstly, the reachable sets obtained from a nonconvex initial set using both the nominal system model and Algorithm~\ref{alg:offline-online} are illustrated in Fig.~\ref{fig:projSetA}. The reachable set $\tilde{\mathcal{R}}_k^{\mathrm{nc}}$ corresponds to the one computed using the model set $\mathcal{M}_0^\Sigma$, without applying any set refinement. $\hat{\mathcal{R}}_k^{\mathrm{nc}}$ represents the reachable set computed after performing the first online set refinement using new incoming data through Algorithm~\ref{alg:offline-online}. Owing to the exact set multiplication employed in the algorithm, the nonconvexity of the initial set is preserved throughout the reachability propagation process. Furthermore, Proposition~\ref{prop:multi} guarantees that the initial reachability computation remains applicable to nonconvex initial sets by employing exact multiplication, thus addressing the limitation of the method proposed by Alanwar et al.~\cite{Alanwar2023Datadriven}, which is confined to convex cases.

Secondly, to evaluate the performance of the proposed algorithm, a comparative analysis was conducted under two scenarios. First, trajectory data collected over the interval $[0, T_1]$ with $N_1$ trajectories were used to perform the offline initialization step of Algorithm~\ref{alg:offline-online}. At time $T_1$, new data over $[T_1, T_2]$ containing $N_2$ trajectories became available, triggering the online set refinement step of Algorithm~\ref{alg:offline-online} and yielding the updated reachable set $\hat{\mathcal{R}}_k^{\mathrm{nc}}$. For consistency and fair comparison with the method of Alanwar et al.~\cite{Alanwar2023Datadriven}, we set $N_1 = N_2$. The reachable set $\bar{\mathcal{R}}_k^{\mathrm{nc}}$ computed using the method of Alanwar et al.~\cite{Alanwar2023Datadriven} is based on the combined dataset $D_a$ over the full interval $[0, T_2]$.

As illustrated in Fig.~\ref{fig:projSetB}, the reachable set $\hat{\mathcal{R}}_k^{\mathrm{c}}$ computed using the Algorithm~\ref{alg:offline-online} is less conservative than $\bar{\mathcal{R}}_k^{\mathrm{c}}$ obtained using the method of Alanwar et al.~\cite{Alanwar2023Datadriven}. This improvement stems from the incremental refinement of the model set, where newly acquired data are used to further constrain the previously computed set of models. As a result, a more accurate representation of system behavior is achieved in the form of a CPMZ, rather than a CMZ.
The use of exact set multiplication between CPMZ and CPZ in Algorithm~\ref{alg:offline-online} ensures exact set propagation, thereby further reducing conservativeness.
These results highlight the advantages of the proposed incremental approach over direct one-shot data-driven methods, demonstrating its improved ability to capture system dynamics and perform set propagation.

\subsection{Polynomial Systems}
We consider the problem of computing the reachable sets of the nonlinear discrete-time system described by
\begin{align}
    f_p(x,u) &= \begin{bmatrix}
    0.7x_1 + u_1 + 0.32x_1^2\\
    0.09x_1 + 0.32 u_2x_1 + 0.4 x_2^2
    \end{bmatrix}.
    \label{eq:pnonlinearexample}
\end{align}
The input set is $\mathcal{U}_k=\bigzono{\begin{bmatrix} 0.2 & 0.3 \end{bmatrix}^\top,\text{diag}(\begin{bmatrix} 0.01 & 0.02 \end{bmatrix})}$. We consider computing the reachable sets under zonotopic disturbances of
different magnitudes to enable a clear comparison of the algorithmic behavior
across noise levels. We used as input data 140 data points ($T=140$) from 20 trajectories, of length seven.

Fig.~\ref{fig:projSetB1}(a) illustrates the model-based reachable sets ${\mathcal{R}}^{p,\mathrm{c}}_k$ and ${\mathcal{R}}^{p,\mathrm{nc}}_k$ computed
using Algorithm~\ref{alg:PolyReachability}.
The disturbance is modeled as a zonotope
$\mathcal{Z}_w=\bigzono{\mathbf{0},\,0.7\times10^{-4}\mathbf{1}}$, and two types of
initial sets are considered, namely a convex zonotope $
\mathcal{X}_0^{p,\mathrm{c}}
=
\bigzono{
\begin{bmatrix} 1 & 1.6 \end{bmatrix}^\top,
\;
2\,\mathrm{diag}\left(\begin{bmatrix} 0.05 & 0.1 \end{bmatrix}\right)
}.
$ and a nonconvex 
CPZ $
\mathcal{X}_0^{p,\mathrm{nc}}
=
\bigzono{
\begin{bmatrix}
1 \\
2.2
\end{bmatrix},
\;
0.1\,I_2,
\;
\begin{bmatrix}
2 & 1 \\
1 & 2
\end{bmatrix}
}_{\text{CPZ}}.
$
The reachable sets ${\mathcal{R}}^{p,\mathrm{c}}_k$ and ${\mathcal{R}}^{p,\mathrm{nc}}_k$ are propagated using the exact polynomial reachable set
computation procedure described in Algorithm~\ref{alg:PolyReachability}.
For comparison, the reachable sets $\check{\mathcal{R}}^{p,\mathrm{c}}_k$ and $\check{\mathcal{R}}^{p,\mathrm{nc}}_k$ are also computed using the model-based
reachability analysis tools provided by CORA~\cite{althoff2015introduction},
which rely on state-of-the-art over-approximation techniques to handle nonlinear
dynamics.
Due to the nonlinearity of the system, CORA computes an over-approximation of the
reachable set, as exact reachable set computation is generally not supported by
existing model-based methods.
In addition, $500$ random trajectories are generated by simulating the true
polynomial system under random inputs and disturbances.
As shown in Fig.~\ref{fig:projSetB1}(a), all simulated state trajectories are fully
contained within the reachable sets ${\mathcal{R}}^{p,\mathrm{c}}_k$ and ${\mathcal{R}}^{p,\mathrm{nc}}_k$ computed by Algorithm~\ref{alg:PolyReachability}.
Notably, several extreme simulation points lie directly on the boundary of the
computed reachable set, indicating that the obtained set tightly captures the
true reachable states.
This observation suggests that the reachable sets ${\mathcal{R}}^{p,\mathrm{c}}_k$ and ${\mathcal{R}}^{p,\mathrm{nc}}_k$ computed by
Algorithm~\ref{alg:PolyReachability} tightly encloses the simulated trajectories under the chosen uncertainty bounds, indicating that no additional relaxation is introduced during propagation.

Fig.~\ref{fig:projSetB1}(b) illustrates the reachable sets computed for a convex
initial set under a small disturbance, where the noise is modeled as the
zonotope $\mathcal{Z}_w=\bigzono{\mathbf{0},\,0.7\times10^{-5}\mathbf{1}}$.
The reachable sets $\hat{\mathcal{R}}_k^{p,\mathrm{c}}$ obtained using Algorithm~\ref{alg:DDPolyReachability_update_intersect}
are compared with those $\bar{\mathcal{R}}_k^{p,\mathrm{c}}$ computed using the data-driven polynomial reachability
method proposed by Alanwar et al.~\cite{Alanwar2023Datadriven}.
As shown in the figure, the reachable sets $\hat{\mathcal{R}}_k^{p,\mathrm{c}}$ computed by
Algorithm~\ref{alg:DDPolyReachability_update_intersect} are significantly less
conservative than those $\bar{\mathcal{R}}_k^{p,\mathrm{c}}$ obtained using the interval-based approach of
\cite{Alanwar2023Datadriven}.
In particular, under small disturbance levels, the reachable sets $\hat{\mathcal{R}}_k^{p,\mathrm{c}}$ produced by
the proposed data-driven algorithm closely match the model-based reachable sets ${\mathcal{R}}_k^{p,\mathrm{c}}$
computed using Algorithm~\ref{alg:PolyReachability}.
This observation indicates that the proposed method is capable of accurately
capturing the system dynamics from data and performing tight forward
propagation without introducing unnecessary conservatism.
Overall, the results in Fig.~\ref{fig:projSetB}(b) demonstrate the effectiveness
of the proposed data-driven polynomial reachability framework, especially in
low-noise scenarios, and highlight its advantage over existing interval-based
data-driven approaches.

Fig.~\ref{fig:projSetB1}(c) illustrates the reachable sets $\tilde{\mathcal{R}}^{p,\mathrm{nc}}_k$ and $\hat{\mathcal{R}}_k^{p,\mathrm{nc}}$ computed for a nonconvex
initial set ${\mathcal{X}}_0^{p,\mathrm{nc}}$ under a larger disturbance level, where the noise is modeled as the
zonotope $\mathcal{Z}_w=\bigzono{\mathbf{0},\,7\times10^{-3}\mathbf{1}}$.
The reachable sets $\tilde{\mathcal{R}}^{p,\mathrm{nc}}_k$ and $\hat{\mathcal{R}}_k^{p,\mathrm{nc}}$ are computed using
Algorithm~\ref{alg:DDPolyReachability_update_intersect}, and the nonconvex initial
set ${\mathcal{X}}_0^{p,\mathrm{nc}}$ is represented as a CPZ.
As in the previous experiments, the input-state data consist of $140$ data points
($T=140$) collected from $20$ trajectories of length seven.
To evaluate the effect of online model refinement, the data are divided into two
segments.
In the first stage, data collected over the interval $[0,T_1]$ with $N_1$
trajectories are used to perform the offline initialization step of
Algorithm~\ref{alg:DDPolyReachability_update_intersect} and yielding initial reachable sets $\tilde{\mathcal{R}}^{p,\mathrm{nc}}_k$.
At time $T_1$, additional data over the interval $[T_1,T_2]$ with $N_2$ trajectories
become available, triggering the online set refinement step and yielding
updated reachable sets $\hat{\mathcal{R}}_k^{p,\mathrm{nc}}$.
For consistency, we set $N_1 = N_2$, corresponding to an equal partition of the
available trajectories.
As shown in Fig.~\ref{fig:projSetB1}(c), under the larger disturbance level, the
reachable sets $\tilde{\mathcal{R}}^{p,\mathrm{nc}}_k$ computed prior to model refinement are relatively conservative.
After performing a single set refinement step using newly acquired data, the
reachable sets $\hat{\mathcal{R}}_k^{p,\mathrm{nc}}$ become significantly tighter.
This demonstrates that the proposed algorithm is capable of effectively reducing
conservatism through online model refinement even in the presence of
large-magnitude disturbances.
Moreover, the results confirm that
Algorithm~\ref{alg:DDPolyReachability_update_intersect} can accurately propagate
complex nonconvex reachable sets while maintaining a zonotopic representation,
highlighting its robustness and flexibility in challenging polynomial systems.

Table~\ref{tab:timing} reports the runtime in seconds for computing the first-step reachable set in Fig.~\ref{fig:x3x4_b} and Fig.~\ref{fig:x3x4_c}, in order to avoid the influence of reduction operations that may be applied during multi-step set propagation.
 In Fig.~\ref{fig:x3x4_b}, employing exact multiplication incurs a modest additional computational cost compared with the baseline and yields a visibly tighter reachable set. In Fig.~\ref{fig:x3x4_c}, the model-based computation is slower than the approach of Alanwar et al.~\cite{Alanwar2023Datadriven}, since Algorithm~\ref{alg:PolyReachability} computes an exact reachable set for the polynomial dynamics. Although the runtime increases more substantially in this case, this is primarily due to the use of exact set propagation via the proposed data-driven exact multiplication, and it yields a markedly less conservative reachable set, as reflected by the tighter reachable sets. Overall, the timing results quantify the expected trade-off: employing exact multiplication in the data-driven propagation increases runtime, but can substantially reduce conservativeness in the computed reachable sets.

\begin{table}[t]
\centering
\begin{tabular}{lccc}
\hline
Scenario & Model Based & Exact Multiplication & Alanwar et al.~\cite{Alanwar2023Datadriven} \\
\hline
Fig.~\ref{fig:x3x4_b} & 0.0033523 & 0.022212  & 0.0059221 \\
Fig.~\ref{fig:x3x4_c} & 0.0081695 & 1.999900  & 0.0044745 \\
\hline
\end{tabular}
\caption{Comparison of reachable set computation times (in seconds).}
\label{tab:timing}
\end{table}

\section{Conclusion} \label{sec:conclusion}


This paper investigated data-driven reachability analysis for systems with unknown dynamics, with a particular focus on the exact propagation of nonconvex reachable sets. The main difficulty addressed is the lack of closure of commonly used set representations under the multiplication between data-consistent model sets and nonconvex state sets, which inevitably leads to conservative over-approximations.
To resolve this issue, we introduced CPMZs and showed that their multiplication with CPZs admits an algebraically exact CPZ representation. This property enabled an exact set-based time-update operation that preserves the nonconvex geometry of reachable sets without relying on relaxation or enclosure techniques. Since CMZs are a special case of CPMZs, the proposed framework naturally generalizes existing data-driven reachability methods for linear systems.
The framework further supports iterative refinement of data-consistent model sets via set intersection, enabling the online incorporation of newly acquired data and yielding progressively tighter reachable sets. Moreover, the proposed algebraic structure enables algebraically exact set propagation for polynomial systems within CPZ arithmetic by evaluating monomial dynamics directly at the set level through exact CPZ multiplication and exact Cartesian products, thereby avoiding over-approximations.

Future work will focus on extending the proposed framework to Lipschitz nonlinear systems, where nonlinear dynamics admit bounded incremental behavior but are not polynomial. Integrating Lipschitz bounds into the proposed exact or minimally conservative set propagation framework remains an important direction for further research.

\bibliographystyle{IEEEtran}
\bibliography{ref}

@article{alanwar2025polynomial,
  title={Polynomial logical zonotope: A set representation for reachability analysis of logical systems},
  author={Alanwar, Amr and Jiang, Frank J and Johansson, Karl H},
  journal={Automatica},
  volume={171},
  pages={111896},
  year={2025},
  publisher={Elsevier}
}

@article{Alanwar2023Datadriven,
    title = {{Data-driven reachability analysis from noisy data}},
    year = {2023},
    journal = {IEEE Transactions on Automatic Control},
    author = {Alanwar, Amr and Koch, Anne and Allgower, Frank and Johansson, Karl Henrik},
    number = {5},
    pages = {3054--3069},
    volume = {68}
}

@inproceedings{alanwar2021data,
  title={Data-driven reachability analysis using matrix zonotopes},
  author={Alanwar, Amr and Koch, Anne and Allg{\"o}wer, Frank and Johansson, Karl Henrik},
  booktitle={Learning for Dynamics and Control},
  pages={163--175},
  year={2021}
}

@inproceedings{alanwar2022enhancing,
  title={Enhancing data-driven reachability analysis using temporal logic side information},
  author={Alanwar, Amr and Jiang, Frank J and Sharifi, Maryam and Dimarogonas, Dimos V and Johansson, Karl H},
  booktitle={International Conference on Robotics and Automation (ICRA)},
  pages={6793--6799},
  year={2022}
}

@inproceedings{choi2023data,
  title={Data-driven forward stochastic reachability analysis for human-in-the-loop systems},
  author={Choi, Joonwon and Byeon, Sooyung and Hwang, Inseok},
  booktitle={62nd IEEE Conference on Decision and Control (CDC)},
  pages={1730--1735},
  year={2023}
}

@article{devonport2023data,
  title={Data-driven reachability and support estimation with {Christoffel} functions},
  author={Devonport, Alex and Yang, Forest and El Ghaoui, Laurent and Arcak, Murat},
  journal={IEEE Transactions on Automatic Control},
  volume={68},
  number={9},
  pages={5216--5229},
  year={2023},
  publisher={IEEE}
}

@inproceedings{dietrich2024nonconvex,
  title={Nonconvex scenario optimization for data-driven reachability},
  author={Dietrich, Elizabeth and Devonport, Alex and Arcak, Murat},
  booktitle={6th Annual Learning for Dynamics \& Control Conference},
  pages={514--527},
  year={2024}
}

@article{djeumou2022fly,
  title={On-the-fly control of unknown systems: From side information to performance guarantees through reachability},
  author={Djeumou, Franck and Vinod, Abraham P and Goubault, Eric and Putot, Sylvie and Topcu, Ufuk},
  journal={IEEE Transactions on Automatic Control},
  volume={68},
  number={8},
  pages={4857--4872},
  year={2022},
  publisher={IEEE}
}

@inproceedings{govindarajan2017data,
  title={Data-driven reachability analysis for human-in-the-loop systems},
  author={Govindarajan, Vijay and Driggs-Campbell, Katherine and Bajcsy, Ruzena},
  booktitle={56th IEEE Conference on Decision and Control (CDC)},
  pages={2617--2622},
  year={2017}
}

@inproceedings{griffioen2023data,
  title={Data-Driven Reachability Analysis for Gaussian Process State Space Models},
  author={Griffioen, Paul and Arcak, Murat},
  booktitle={62nd IEEE Conference on Decision and Control (CDC)},
  pages={4100--4105},
  year={2023}
}

@article{kochdumper2020sparse,
  title={Sparse polynomial zonotopes: A novel set representation for reachability analysis},
  author={Kochdumper, Niklas and Althoff, Matthias},
  journal={IEEE Transactions on Automatic Control},
  volume={66},
  number={9},
  pages={4043--4058},
  year={2020},
  publisher={IEEE}
}

@article{kochdumper2023constrained,
  title={Constrained polynomial zonotopes},
  author={Kochdumper, Niklas and Althoff, Matthias},
  journal={Acta Informatica},
  volume={60},
  number={3},
  pages={279--316},
  year={2023},
  publisher={Springer}
}

@article{park2024data,
  title={Data-driven Reachability Analysis for Nonlinear Systems},
  author={Park, Hyunsang and Vijay, Vishnu and Hwang, Inseok},
  journal={IEEE Control Systems Letters},
  year={2024},
  publisher={IEEE}
}

@article{scott2016constrained,
  title={Constrained zonotopes: A new tool for set-based estimation and fault detection},
  author={Scott, Joseph K and Raimondo, Davide M and Marseglia, Giuseppe Roberto and Braatz, Richard D},
  journal={Automatica},
  volume={69},
  pages={126--136},
  year={2016},
  publisher={Elsevier}
}

@article{thorpe2019model,
  title={Model-free stochastic reachability using kernel distribution embeddings},
  author={Thorpe, Adam J and Oishi, Meeko MK},
  journal={IEEE Control Systems Letters},
  volume={4},
  number={2},
  pages={512--517},
  year={2019},
  publisher={IEEE}
}

@inproceedings{thorpe2021sreachtools,
  title={{SReachTools} kernel module: Data-driven stochastic reachability using {Hilbert} space embeddings of distributions},
  author={Thorpe, Adam J and Ortiz, Kendric R and Oishi, Meeko MK},
  booktitle={60th IEEE Conference on Decision and Control (CDC)},
  pages={5073--5079},
  year={2021}
}

@inproceedings{wang2023data,
  title={Data-driven reachability analysis of Lipschitz nonlinear systems via support vector data description},
  author={Wang, Zheming and Chen, Bo and Jungers, Rapha{\"e}l M and Yu, Li},
  booktitle={62nd IEEE Conference on Decision and Control (CDC)},
  pages={7043--7048},
  year={2023}
}

@phdthesis{Althoff2010PhD,
    title = {{Reachability Analysis and its Application to the Safety Assessment of Autonomous Cars}},
    year = {2010},
    booktitle = {Fakult{\"{a}}t f{\"{u}}r Elektrotechnik und Informationstechnik},
    author = {Althoff, Matthias},
    pages = {221},
    school = {Technische Universit{\"{a}}t M{\"{u}}nchen},
}

@article{Kuhn1998,
    title = {{Rigorously computed orbits of dynamical systems without the wrapping effect}},
    year = {1998},
    journal = {Computing},
    author = {K{\"{u}}hn, W.},
    number = {1},
    month = {3},
    pages = {47--67},
    volume = {61},
}

@article{farjadnia2024robust,
  title={Robust Data-Driven Tube-Based Zonotopic Predictive Control with Closed-Loop Guarantees},
  author={Farjadnia, Mahsa and Fontan, Angela and Alanwar, Amr and Molinari, Marco and Johansson, Karl Henrik},
  journal={arXiv preprint arXiv:2409.14366},
  year={2024}
}

@book{boyd2004convex,
  title={Convex optimization},
  author={Boyd, Stephen and Vandenberghe, Lieven},
  year={2004},
  publisher={Cambridge university press}
}

@inproceedings{luo2023reachability,
  title={Reachability analysis for linear systems with uncertain parameters using polynomial zonotopes},
  author={Luo, Ertai and Kochdumper, Niklas and Bak, Stanley},
  booktitle={Proceedings of the 26th ACM International Conference on Hybrid Systems: Computation and Control},
  pages={1--12},
  year={2023}
}

@inproceedings{sivaramakrishnan2024stochastic,
  title={Stochastic reachability of uncontrolled systems via probability measures: Approximation via deep neural networks},
  author={Sivaramakrishnan, Karthik and Sivaramakrishnan, Vignesh and Devonport, Rosalyn A and Oishi, Meeko MK},
  booktitle={2024 IEEE 63rd Conference on Decision and Control (CDC)},
  pages={7534--7541},
  year={2024},
  organization={IEEE}
}

@article{hashemi2024statistical,
  title={Statistical reachability analysis of stochastic cyber-physical systems under distribution shift},
  author={Hashemi, Navid and Lindemann, Lars and Deshmukh, Jyotirmoy V},
  journal={IEEE Transactions on Computer-Aided Design of Integrated Circuits and Systems},
  volume={43},
  number={11},
  pages={4250--4261},
  year={2024},
  publisher={IEEE}
}

@inproceedings{althoff2015introduction,
  title={An introduction to CORA 2015},
  author={Althoff, Matthias},
  booktitle={Proc. of the workshop on applied verification for continuous and hybrid systems},
  pages={120--151},
  year={2015}
}

@book{rierson2017developing,
  title={Developing safety-critical software: a practical guide for aviation software and DO-178C compliance},
  author={Rierson, Leanna},
  year={2017},
  publisher={CRC Press}
}

@article{conf:polyDissipat,
  title={Dissipativity verification with guarantees for polynomial systems from noisy input-state data},
  author={Martin, Tim and Allg{\"o}wer, Frank},
  journal={IEEE Control Systems Letters},
  volume={5},
  number={4},
  pages={1399--1404},
  year={2020}
}

@article{dang2012reachability,
  title={Reachability Analysis for Polynomial Dynamical Systems Using the Bernstein Expansion.},
  author={Dang, Thao and Testylier, Romain},
  journal={Reliab. Comput.},
  volume={17},
  number={2},
  pages={128--152},
  year={2012}
}

@article{dreossi2017reachability,
  title={Reachability computation for polynomial dynamical systems},
  author={Dreossi, Tommaso and Dang, Thao and Piazza, Carla},
  journal={Formal Methods in System Design},
  volume={50},
  number={1},
  pages={1--38},
  year={2017},
  publisher={Springer}
}

@article{lin2022reachable,
  title={Reachable set estimation and safety verification of nonlinear systems via iterative sums of squares programming},
  author={Lin, Wang and Yang, Zhengfeng and Ding, Zuohua},
  journal={Journal of Systems Science and Complexity},
  volume={35},
  number={3},
  pages={1154--1172},
  year={2022},
  publisher={Springer}
}

@article{jones2019using,
  title={Using SOS and sublevel set volume minimization for estimation of forward reachable sets},
  author={Jones, Morgan and Peet, Matthew M},
  journal={IFAC-PapersOnLine},
  volume={52},
  number={16},
  pages={484--489},
  year={2019},
  publisher={Elsevier}
}

@inproceedings{marechal2017efficient,
  title={Efficient elimination of redundancies in polyhedra by raytracing},
  author={Mar{\'e}chal, Alexandre and P{\'e}rin, Micha{\"e}l},
  booktitle={International Conference on Verification, Model Checking, and Abstract Interpretation},
  pages={367--385},
  year={2017},
  organization={Springer}
}

\end{document}